\documentclass[10pt]{article} 
\usepackage{amsthm}
\theoremstyle{definition}
\newtheorem{defi}{Definition}[section]
\newtheorem{remark}[defi]{Remark}

\newtheorem{example}[defi]{Example}

\theoremstyle{plain}
\newtheorem{theorem}[defi]{Theorem}

\newtheorem{lemma}[defi]{Lemma}
\newtheorem{prop}[defi]{Proposition}

\newtheorem{asum}[defi]{Assumption}
\usepackage{cite}
\usepackage{amsmath, amsfonts}
\usepackage[english]{babel}
\usepackage{enumerate}
\usepackage [T1]{fontenc}
\usepackage{endnotes}
\usepackage{xcolor}
\usepackage{abstract}
\usepackage{indentfirst}
\usepackage{enumitem,blindtext}
\usepackage[paper=a4paper,left=30mm,right=30mm,top=30mm,bottom=30mm]{geometry}

\numberwithin{equation}{section}

\DeclareMathOperator*{\esssup}{ess\,su{p}^{\textit{P}}}
\DeclareMathOperator*{\esssupT}{ess\,su{p}^{\textit{$\tilde{P}$}}}
\DeclareMathOperator*{\esssupO}{ess\,su{p}^{\textit{$\overline{P}$}}}
\DeclareMathOperator*{\essinf}{ess\,inf^{\textit{$\tilde{P}$}}}

\title{Reduced-form setting under model uncertainty with non-linear affine intensities}
\author{Francesca Biagini\footnote{Department of Mathematics, Workgroup Financial and Insurance Mathematics, University of Munich (LMU), Theresienstra\ss e 39, 80333 Munich, Germany. Email: francesca.biagini@math.lmu.de} \quad Katharina Oberpriller\footnote{Department of Mathematics of Natural, Social and Life Sciences, Gran Sasso Science Institute (GSSI), Viale F. Crispi 7, 67100 L'Aquila, Italy. Email: katharina.oberpriller@gssi.it }}
\begin{document}
\maketitle

\begin{abstract}
	In this paper we extend the reduced-form setting under model uncertainty introduced in \cite{bz_2019} to include intensities following an affine process under parameter uncertainty, as defined in \cite{fns_2019}. This framework allows to introduce a longevity bond under model uncertainty in a consistent way with the classical case under one prior, and to compute its valuation numerically. Moreover, we are able to price a contingent claim with the sublinear conditional operator such that the extended market is still arbitrage-free in the sense of ``No Arbitrage of the first kind'' as in \cite{bbkn_2017}.
	\end{abstract}

\begin{section}{Introduction}
Aim of this paper is to extend the reduced-form setting under model uncertainty as introduced in \cite{bz_2019} to include non-linear affine intensities as defined in \cite{fns_2019}. In this way we are able to introduce a longevity bond under model uncertainty in an arbitrage-free way and numerically compute its value process in several examples. Furthermore, we apply these results to the arbitrage-free pricing of a general contingent claim under model uncertainty. \\
More precisely, in \cite{bz_2019} the classical reduced-form framework as in \cite{bielecki_rutkowski_2004} is extended under model uncertainty by defining a sublinear conditional operator with respect to a progressively enlarged filtration $\mathbb{G}$ and a family of probability measures possibly mutually singular to each other, which is an extension of the sublinear conditional operator with respect to $\mathbb{F}$ introduced in \cite{nh_2013}. In this setting no specific structure or assumptions are made for the intensity process. In the last few years several papers dealing with short rate modeling under model uncertainty have been published, e.g., \cite{hoelzermann_2020}, \cite{hoelzermann_2019}, \cite{hoelzermann_quian_2020}, \cite{fadina_schmidt_2019}, \cite{fns_2019}. A more general approach is treated in \cite{fns_2019} by considering affine processes under parameter uncertainty, called non-linear affine processes, as an extension of the non-linear L\'evy processes in \cite{neufeld_nutz_2016}. More specifically, one-dimensional non-linear affine processes are defined as a family of semimartingale laws whose differential characteristics are bounded from above and below by affine functions of the current states. In financial applications, affine processes are not only relevant for short rate models but also for modeling the stochastic mortality/default intensity, e.g., in \cite{dahl_2004}, \cite{b_2004}, \cite{s_2006} and \cite{luciano_vigna_2005}, as they allow analytically tractable models.\\
Here we wish to provide the most general reduced-form setting under model uncertainty which allows numerical tractability or explicit computation for pricing insurance liabilities or credit derivatives. Hence we extend the results in \cite{bz_2019} by representing the mortality intensity as non-linear affine processes in the sense of \cite{fns_2019}. By doing so we are able to construct a general market model, where the risky assets are local $\mathbb{F}$-martingales and the intensity process is a non-linear affine process under the considered (time-dependent increasing) families of probability measures. The associated sublinear conditional operator can then be used to evaluate insurance products by taking into account the (non-linear) affine structure of the mortality intensity. Furthermore, we give some examples for families of probability measures such that the market model satisfies the required assumptions. From a mathematical point of view the construction of the sublinear conditional operator requires some regularity assumptions for the families of probability measures as in \cite{nh_2013}. In our context the difficulty lies in constructing families of priors which satisfies these assumptions as well as the desired properties concerning the market model and the affine structure of the intensity. \\
In general, the mortality intensity is used to define the survival index and can be seen as building block for mortality linked securities \cite{cbd_2006}. These kind of financial instruments which started appearing on the market around 2003 have the aim to reduce the mortality and longevity risk connected to life insurance and pension products. One of the basic products of this type are longevity bonds which pay the survivor index at the maturity and it is common to price them with the risk-neutral measure such that the extended market including the longevity bond is arbitrage-free \cite{cbd_2006}. In this work we are able to introduce the definition of a longevity bond under model uncertainty in a consistent way with the classical setting under one prior. To this purpose, we use the sublinear conditional operator of \cite{bz_2019}. As already mentioned in \cite{bz_2019}, the sublinear conditional operator is a priori not c\`{a}dl\`{a}g which can lead to problems as c\`{a}dl\`{a}g paths are a common standard assumption. This problem is solved in \cite{bz_2019} by considering a fixed set of probability measures. Here this result does not hold because we need to work with time-dependent increasing families of priors in order to include non-linear affine intensities. Nevertheless, we are able to find conditions such that there exists a c\`{a}dl\`{a}g modification for the conditional sublinear operator. By generalizing the representation of the sublinear expectation with Riccati equations from \cite{fns_2019}, we are able to numerically compute the value of a longevity bond for some relevant examples. Moreover, numerical computations can also be used for the valuation of general endowment contracts under an independence assumption between the asset's price process and the mortality intensity.\\
Motivated by the valuation of the longevity bond, we examine if the sublinear conditional operator in \cite{bz_2019} can be used for pricing a contingent claim under model uncertainty such that the extended market is arbitrage-free. To do so, we first need to choose an appropriate definition of arbitrage in a continuous time setting under model uncertainty. While for the discrete time setting there exists a broad literature about no arbitrage and related concepts under model uncertainty, e.g., \cite{acciaio_beiglboeck_penkner_schachermayer_2013}, \cite{bayraktar_zhang_zhou_2014}, \cite{bouchard_nutz_2015} and \cite{nutz_superreplication_discrete_2014}, the situation is different for the continuous case. In \cite{vorbrink_2014} no-arbitrage is studied within a setting of volatility uncertainty. In \cite{bbkn_2017} they introduce a robust version of arbitrage of the first kind and derive the fundamental theorem of asset pricing. By applying this definition to our setting, we show that the extended sublinear operator can be used to price a contingent claim such that the extended market allows no arbitrage of the first kind under model uncertainty as in \cite{bbkn_2017}. This result requires assumptions about the trading strategies which are however not restrictive in an insurance setting. Moreover, we discuss the relation of this valuation to the superhedging price of a contingent claim under model uncertainty given in \cite{bz_2019}.  \\
The paper is organized as follows. In Section \ref{Section_Reduced_framework} we outline the setting in \cite{bz_2019} and extend the definition of the sublinear conditional operator with respect to time-dependent increasing families of probability measures instead of a fixed set as in the original framework. In Section \ref{SectionAffineProcesses} we introduce the definition of a non-linear affine process defined as in \cite{fns_2019}. Next, we define a market model under uncertainty combining the two settings and illustrate this with some examples in Section \ref{Section_Combination}. In Section \ref{SectionMainLongevity} we give the definition of a longevity bond under model uncertainty and derive its numerical approximation via Riccati equations. Moreover, we show under which conditions it is possible to find a c\`{a}dl\`{a}g quasi-sure modification of the sublinear conditional operator and prove that these assumptions are satisfied in the examples given in Section \ref{Section_Combination}. In Section \ref{SectionPricing} we introduce the definition of no arbitrage under first kind in our framework and study arbitrage-free pricing of a contingent claim via the sublinear conditional operator.
\end{section}

\begin{section}{Reduced-form setting under model uncertainty} \label{Section_Reduced_framework}
A reduced-form setting for credit and insurance markets under model uncertainty is introduced in the paper \cite{bz_2019} by defining a sublinear conditional operator with respect to a progressively enlarged filtration and a family of probability measures possibly mutually singular to each other. This is not a straight forward extension of the construction in \cite{nh_2013}, because the approach in \cite{nh_2013} relies on special properties of the natural filtration generated by the canonical process which are not any longer satisfied by the enlarged filtration. For example, the Galmarino's test\footnote{Galmarino's Test \cite[Exercise 4.21]{revuz_yor_2005}: Let $\Omega=C(\mathbb{R}_+,\mathbb{R})$, $\mathcal{F}$ the Borel $\sigma$-algebra with respect to the topology of locally uniform convergence and $\mathbb{F}$ be the raw filtration generated by the canonical process $B$ on $\Omega$. Then a $\mathcal{F}$-measurable function $\tau: \Omega \to \mathbb{R}_+$ is a $\mathbb{F}$-stopping time if and only if $\tau(\omega) \leq t$ and $\omega\vert_{[0,t]}=\omega'\vert_{[0,t]}$ imply $\tau(\omega)=\tau(\omega')$. Furthermore, given a $\mathbb{F}$-stopping time $\tau$, and $\mathcal{F}$-measurable function $f$ is $\mathcal{F}_{\tau}$-measurable if and only if $f=f \circ \iota_{\tau}$, where $\iota_{\tau}:\Omega \to \Omega$ is the stopping map $(\iota_{\tau}(\omega))_t=\omega_{t \wedge \tau(\omega)}$. } 
	cannot be used in this extended framework, as the assumptions under which it holds are not satisfied in the enlarged filtration.
In the following, we recall the approach in \cite{bz_2019} in a more general version by taking into account families of probability measures $(\mathcal{P}(t,\omega))_{(t,\omega) \in [0,T] \times \Omega}$ on a space $\Omega$ as in \cite{nh_2013} instead of a fixed set $\mathcal{P}$ on $\Omega$ as in \cite{bz_2019}.\\
Fix $T>0$ and consider the space $\Omega=C_0([0,T], \mathbb{R})$ of continuous functions $\omega=(\omega_t)_{t \in [0,T]}$ in $\mathbb{R}$ starting from zero, which is equipped with the topology of locally uniform convergence and is therefore a Polish space. The Borel $\sigma$-algebra on this space is given by $\mathcal{F}=\mathcal{B}(\Omega)$ and the set of probability measures on $(\Omega, \mathcal{F})$ by $\mathcal{P}(\Omega)$. We assume that $\mathcal{P}(\Omega)$ is endowed with the topology of weak convergence. Furthermore, we denote by $B:=(B_t)_{t \in [0,T]}$ the canonical process, i.e., $B_t(\omega)=\omega_t, \ t \in [0,T]$ and its corresponding raw filtration by $\mathbb{F}:=(\mathcal{F}_t)_{t \in [0,T]}$ with $\mathcal{F}_0=\lbrace \emptyset, \Omega \rbrace$ and $\mathcal{F}_{T}:=\bigvee_{t \in [0,T]} \mathcal{F}_t=\mathcal{F}$. For every given $P \in \mathcal{P}(\Omega)$ and $t \in [0,T]$, we define $\mathcal{N}_t^P$ as the collection of sets which are $(P, \mathcal{F}_t)$-null and we consider the following filtration $\mathbb{F}^*:=(\mathcal{F}^*_t)_{t \in [0,T]}$ defined by
\begin{equation}
	\mathcal{F}_t^*:=\mathcal{F}_t \vee \mathcal{N}_t^*, \quad \mathcal{N}_t^*:=\bigcap_{P \in \mathcal{P}(\Omega)} \mathcal{N}_t^P. \label{filtration}
\end{equation}
For a given family of probability measures $\mathcal{P}$ on $\Omega$ we define the $\sigma$-algebra $\mathcal{F}_T^{\mathcal{P}}$ by 
\begin{equation} \label{defiNullsets}
	\mathcal{F}_T^{\mathcal{P}}:=\mathcal{F} \vee \mathcal{N}_T^{\mathcal{P}}, \quad \mathcal{N}_T^{\mathcal{P}}:=\bigcap_{P \in \mathcal{P}} N_T^P
\end{equation}
and the filtration $\mathbb{F}^{*,\mathcal{P}}:=(\mathcal{F}_t^{*,\mathcal{P}})_{t \in [0,T]}$  is given by
\begin{equation}
	\mathcal{F}_t^{*,\mathcal{P}}:=\mathcal{F}_t^* \vee \mathcal{N}_T^{\mathcal{P}}, \quad t \in [0,T], \label{filtrationP}
\end{equation}
where $\mathcal{N}_T^{\mathcal{P}}$ is the collection of sets which are $(P, \mathcal{F}_T)$-null for all $P \in \mathcal{P}$.
We follow the approach of \cite{nh_2013} for defining sublinear expectations and introduce the following notation. Let $\tau$ be a $[0,T]$-valued $\mathbb{F}$-stopping time and $\omega \in \Omega$. For every $\tilde{\omega} \in \Omega$, the concatenation process $\omega \otimes_{\tau} \tilde{\omega}:=((\omega \otimes_{\tau} \tilde{\omega})_t)_{t \in [0,T]}$ of $(\omega, \omega')$ at $\tau$ is given by
\begin{equation}
	(\omega \otimes_{\tau} \tilde{\omega})_t:= \omega_t \textbf{1}_{[0,\tau(\omega))}(t)+(\omega_{\tau(\omega)}+ \tilde{\omega}_{t-\tau(\omega)})\textbf{1}_{[\tau(\omega)),T]}(t), \quad t \in [0,T]. \label{concatenation}
\end{equation}
Furthermore, for every function $X$ on $\Omega$ we define the function $X^{\tau,\omega}$ on $\Omega$ by
\begin{equation}
	X^{\tau,\omega}(\tilde{\omega}):=X(\omega \otimes_{\tau} \tilde{\omega}), \quad \tilde{\omega} \in \Omega. \label{concaRV}
\end{equation}
Given a probability measure $P \in \mathcal{P}(\Omega)$ and the regular conditional probability distribution $P^{\omega}_{\tau}$ of $P$ given $\mathcal{F}_{\tau}$, we consider the probability measure $P^{\tau, \omega} \in \mathcal{P}(\Omega)$ given by
\begin{equation}
	P^{\tau, \omega}(A):=P_{\tau}^{\omega}(\omega \otimes_{\tau} A), \quad A \in \mathcal{F}, \label{condProbability} 
\end{equation}
with $\omega \otimes_{\tau} A=\lbrace \omega \otimes_{\tau} \tilde{\omega}: \tilde{\omega} \in A \rbrace$. Note that $P^{\omega}_{\tau}$ is concentrated on the paths which coincide with $\omega$ up to time $\tau(\omega)$. \\
For any $(s,\omega) \in [0,T] \times \Omega$ we fix the sets $\mathcal{P}(s,\omega) \subseteq \mathcal{P}(\Omega)$ and assume that
\begin{equation*}
	\mathcal{P}(s,\omega) = \mathcal{P}(s,\tilde{\omega}) \quad \text{ if } \omega \vert_{[0,s]} = \tilde{\omega} \vert_{[0,s]}.
\end{equation*}
The set $\mathcal{P}(0,\omega)$ is independent of $\omega$ and from now on denoted by $\mathcal{P}$. For a stopping time $\sigma$ we put $\mathcal{P}(\sigma,\omega):=\mathcal{P}(\sigma(\omega),\omega)$.
\begin{asum} \label{assumptionnutzNew}
	Let $(s,\overline{\omega}) \in [0,T] \times \Omega, P \in \mathcal{P}(s, \overline{\omega})$ and $\tau$ be a stopping time such that $T \geq \tau \geq s$. Set $\eta:=\tau^{s,\overline{\omega}}-s$, then
	\begin{enumerate}
	\itemsep0pt
		\item \emph{Measurability:} The graph $\lbrace (P',\omega): \omega \in \Omega, P' \in \mathcal{P}(\tau,\omega) \rbrace \subseteq \mathcal{P}(\Omega) \times \Omega$ is analytic.
		\item \emph{Invariance:} $P^{\eta, \omega} \in \mathcal{P}(\tau, \overline{\omega} \otimes_s \omega)$ for $P$-a.e. $\omega \in \Omega$.
		\item \emph{Stability under Pasting:} If $\nu: \Omega \to \mathcal{P}(\Omega)$ is an $\mathcal{F}_{\eta}$-measurable kernel and $\nu(\omega) \in \mathcal{P}(\tau, \overline{\omega} \otimes_s \omega)$ for $P$-a.e. $\omega \in \Omega$, then the measure defined by
			\begin{equation}
				\overline{P}(A)=\int \int (\textbf{1}_A)^{\eta,\omega}(\omega')\nu(d\omega';\omega)P(d\omega), \quad A \in \mathcal{F}, \label{OverlineP}
			\end{equation}
			is an element of $\mathcal{P}(s,\overline{\omega})$.
	\end{enumerate}
\end{asum}
The following proposition is the main result in \cite[Theorem 2.3]{nh_2013}.
\begin{prop} \label{SublinearNutz}
	Let Assumption \ref{assumptionnutzNew} hold true, $\sigma \leq \tau \leq T$ be $\mathbb{F}$-stopping times and $X: \Omega \to \overline{\mathbb{R}}$ be an upper semianalytic function on $\Omega$. Then the function $\mathcal{E}_{\tau}(X)$ defined by
		\begin{equation}
			\mathcal{E}_{\tau}(X)(\omega):=\sup_{P \in \mathcal{P}(\tau,\omega)} E^P[X^{\tau,\omega}], \quad \omega \in \Omega, \label{definitionOperator}
		\end{equation}
		is $\mathcal{F}^*_{\tau}$-measurable and upper semianalytic.
		Moreover
		\begin{equation}
			\mathcal{E}_{\sigma}(X)(\omega)= \mathcal{E}_{\sigma}(\mathcal{E}_{\tau}(X))(\omega) \quad \text{ for all } \omega \in \Omega. \label{towerNutz}
		\end{equation}
	Furthermore, the following consistency condition is fulfilled, i.e.,
		\begin{equation}
			\mathcal{E}_{\tau}(X)=\esssup_{P' \in \mathcal{P}(\tau;P)} E^{P'}[X| \mathcal{F}_{\tau}] \quad P\text{-a.s. for all } P \in \mathcal{P}, \label{repesssup}
		\end{equation}
	where $\mathcal{P}(\tau;P)=\lbrace P' \in \mathcal{P}: P'=P \text{ on } \mathcal{F}_{\tau}\rbrace$.
\end{prop}
The family of sublinear conditional expectations $(\mathcal{E}_t)_{t \in [0,T]}$ is called $(\mathcal{P}, \mathbb{F})$-conditional expectation. \\
We now enlarge the underlying space to introduce a random time $\tilde{\tau}$, which is not an $\mathbb{F}$-stopping time but has an $\mathbb{F}$-progressively measurable intensity process $\mu$ to represent a totally unexpected default or decease time under model uncertainty. Let $\hat{\Omega}$ be another Polish space equipped with its Borel $\sigma$-algebra $\mathcal{B}(\hat{\Omega})$. On the product space $(\tilde{\Omega}, \mathcal{G}):=(	\Omega \times \hat{\Omega} , \mathcal{B}(\Omega) \otimes \mathcal{B}(\hat{\Omega}))$ we adopt the following conventions. For every function or process $X$ on $(\Omega,\mathcal{B}(\Omega))$ we denote its natural immersion into the product space by $X(\tilde{\omega}):=X(\omega)$ for all $\omega \in \Omega$ and similarly for processes on $(\tilde{\Omega}, \mathcal{B}(\tilde{\Omega}))$. Furthermore, for every sub-$\sigma$-algebra $\mathcal{A}$ of $\mathcal{B}(\Omega)$, the natural extension as a sub-$\sigma$-algebra of $\mathcal{G}$ on $(\tilde{\Omega}, \mathcal{G})$ is given by $\mathcal{A} \otimes \lbrace \emptyset, \tilde{\Omega} \rbrace$, similarly for sub-$\sigma$-algebras of $\mathcal{B}(\hat{\Omega})$.\\
We fix a probability measure $\hat{P}$ on $(\hat{\Omega},\mathcal{B}(\hat{\Omega}))$ such that $(\hat{\Omega}, \mathcal{B}(\hat{\Omega}), \hat{P})$ is an atomless probability space, i.e., there exists a random variable with an absolutely continuous distribution. Moreover, let $\xi$ be a Borel-measurable surjective random variable
	\begin{equation*}
		\xi:(\hat{\Omega}, \mathcal{B}(\hat{\Omega}), \hat{P}) \to ([0,1], \mathcal{B}([0,1]))
	\end{equation*}
with uniform distribution, i.e., $\xi \in \mathcal{U}([0,1])$. Without loss of generality we assume $\mathcal{B}(\hat{\Omega})=\sigma(\xi)$.\\
The family of all probability measures on $(\tilde{\Omega}, \mathcal{G})$ is denoted by $\mathcal{P}(\tilde{\Omega})$. In particular we are interested in the following families of probability measures $(\tilde{\mathcal{P}}(t,\omega))_{(t,\omega) \in [0,T] \times \Omega}$ with
	\begin{equation}
			\tilde{\mathcal{P}}(t,\omega):=\lbrace \tilde{P} \in \mathcal{P}(\tilde{\Omega}): \tilde{P} = P \otimes \hat{P}, P \in \mathcal{P}(t,\omega) \rbrace \label{probExtendedDependence}
	\end{equation}
for $(t,\omega) \in [0,T] \times \Omega$. As $\mathcal{P}(0,\omega)$ does not depend on $\omega$ this also holds for $\tilde{\mathcal{P}}(0,\omega)$ which is denoted by $\tilde{\mathcal{P}}$.
Moreover, we consider an $\mathbb{R}$-valued, $\mathbb{F}$-adapted, continuous and increasing process $\Gamma:=(\Gamma_t)_{t \geq 0}$ on $(\Omega, \mathcal{B}(\Omega))$ with $\Gamma_0:=0$ and $\Gamma_{\infty}:=+\infty$ such that 
\begin{equation}
	\Gamma_t:=\int_0^t \mu_s ds, \quad t \geq 0, \quad \text{ for all } t \geq 0,  \text{ for all } \omega \in \Omega,  \label{representationIntensity}
\end{equation}
where $\mu=(\mu_t)_{t \geq 0}$ is a nonnegative $\mathbb{F}$-progressively measurable stochastic process with $\int_0^t  \mu_s (\omega) ds < \infty$ for all $t\geq 0, \omega \in \Omega$. 
On $\tilde{\Omega}=\Omega \times \hat{\Omega}$ we define the stopping time $\tilde{\tau}$ by
	\begin{equation}
		\tilde{\tau}=\inf \lbrace t \geq 0: e^{-\Gamma_t} \leq \xi \rbrace = \inf \lbrace t \geq 0: \Gamma_t \geq -\ln \xi \rbrace, \label{definitionTau}
	\end{equation}
where we use the convention $\inf \emptyset = \infty$.\\
We define the filtration $\mathbb{H}:=(\mathcal{H}_t)_{t \in [0,T]}$ on $\tilde{\Omega}$ which is generated by the process $H:=(H_t)_{t \in [0,T]}$ with 
	\begin{equation}
		H_t:=\textbf{1}_{\lbrace \tilde{\tau} \leq t \rbrace}, \quad t \in [0,T], \label{definitionFiltrationH}
	\end{equation}
and consider the enlarged filtration $\mathbb{G}:=(\mathcal{G}_t)_{t \in [0,T]}$ with $\mathcal{G}_t:=\mathcal{F}_t \vee \mathcal{H}_t,$ $ t \in [0,T]$. With this construction it holds $\mathcal{G}=\mathcal{F}_T \otimes \sigma(\xi) = \mathcal{H}_T \vee \mathcal{F}_T=\sigma(\tilde{\tau}) \vee \mathcal{F}_T$.
As in $(\ref{filtration})$ we denote by $\mathbb{G}^*$ the corresponding universally completed filtration. Moreover, let $\mathcal{G}^P:=\mathcal{G} \vee \mathcal{N}_{T}^P$ for $P \in \mathcal{P}$ and $\mathcal{G}^{\mathcal{P}}:=\mathcal{G} \vee \mathcal{N}_{T}^{\mathcal{P}}$ with $\mathcal{N}_{T}^{\mathcal{P}}$ defined in $(\ref{defiNullsets})$.
In addition, we define $L^0(\tilde{\Omega})$ as the space of all $\mathbb{R}$-valued $\mathcal{G}^{\mathcal{P}}$-measurable functions, where we use the following convention. For every $\tilde{P} \in \mathcal{P}(\tilde{\Omega})$, we set $E^{\tilde{P}}[X]:=E^{\tilde{P}}[X^+]-E^{\tilde{P}}[X^-]$ if $E^{\tilde{P}}[X^+]$ or $E^{\tilde{P}}[X^-]$ is finite and $E^{\tilde{P}}[X]:= -\infty$ if $E^{\tilde{P}}[X^+]=E^{\tilde{P}}[X^-]=+\infty$. Furthermore, we introduce the set
\begin{align*}
	L^1({\tilde{\Omega}}):=\lbrace \tilde{X} \vert  \tilde{X}: (\tilde{\Omega}, \mathcal{G}^{\mathcal{P}}) \to (\mathbb{R}, \mathcal{B}(\mathbb{R})) \text{ measurable function such that  }
	\tilde{\mathcal{E}}(\vert \tilde{X} \vert ) < \infty \rbrace.
\end{align*}
Here $\tilde{\mathcal{E}}$ denotes the upper expectation associated to $\tilde{\mathcal{P}}$ defined as
\begin{equation*}
	\tilde{\mathcal{E}}(\tilde{X}):=\sup_{\tilde{P} \in \tilde{\mathcal{P}}} E^{\tilde{P}} [\tilde{X}],  \quad \tilde{X} \in L^0(\tilde{\Omega}).
\end{equation*}
One important step for the main result of \cite{bz_2019} is Proposition 2.13 in\cite{bz_2019}. A similar result has also been derived in Proposition 2.2 in \cite{callegaro_jeanblanc_2018}.
\begin{prop} \label{ExistenceXi}
	Let $t \in [0,T]$. If $\tilde{X}$ is a real-valued $\sigma(\tilde{\tau}) \vee \mathcal{F}_t$-measurable function on $\tilde{\Omega}$, then there exists a unique measurable function 
	\begin{equation*}
		\varphi: (\mathbb{R}_+ \times \Omega, \mathcal{B}(\mathbb{R}_+) \otimes \mathcal{F}_t) \to (\mathbb{R}, \mathcal{B}(\mathbb{R})),
	\end{equation*}
	such that
	\begin{equation}
		\tilde{X}(\omega, \hat{\omega})=\varphi(\tilde{\tau}(\omega, \hat{\omega}), \omega), \quad (\omega, \hat{\omega}) \in \tilde{\Omega}. \label{existenceVarphi}
	\end{equation}
\end{prop}
The existence of such a measurable function $\varphi$ does not depend on the structure of the considered family of probability measures on $\tilde{\Omega}$, as the proof is based on a monotone class argument. Moreover, the other crucial point to extend the sublinear operator to the enlarged space is Proposition \ref{SublinearNutz}. Thus, we are able to state a generalized version of Theorem 2.18, Proposition 2.21 in \cite{bz_2019}.

\begin{prop} \label{extendedOperator}
	Let Assumption \ref{assumptionnutzNew} hold for $(\mathcal{P}(t,\omega))_{(t,\omega) \in [0,T] \times \Omega}$ and consider an upper semianalytic function $\tilde{X}$ on $\tilde{\Omega}$ such that $\tilde{X} \in L^1(\tilde{\Omega})$ or $\tilde{X}$ is $\mathcal{G}^{\mathcal{P}}$-measurable and nonnegative. If $t \in [0,T]$, then the following function 
	\begin{equation}
		\tilde{\mathcal{E}}_t(\tilde{X}):=\textbf{1}_{\lbrace \tilde{\tau} \leq t \rbrace} \mathcal{E}_t(\varphi(x,\cdot))\vert_{x=\tilde{\tau}}+ \textbf{1}_{\lbrace \tilde{\tau}>t \rbrace} \mathcal{E}_t(e^{\Gamma_t} E^{\hat{P}}[\textbf{1}_{\lbrace \tilde{\tau} > t \rbrace } \tilde{X}]) \label{BasicDecomposition}
	\end{equation} 
is well-defined, where $\varphi$ is the measurable function 
	\begin{equation}
		\varphi: (\mathbb{R}_+ \times \Omega, \mathcal{B}(\mathbb{R}_+) \otimes \mathcal{F}_{T}^{\mathcal{P}}) \to (\mathbb{R}, \mathcal{B}(\mathbb{R})), \label{varphi}
	\end{equation}
such that
	\begin{equation}
	\tilde{X}(\omega, \hat{\omega})=\varphi(\tilde{\tau}(\omega, \hat{\omega}), \omega), \quad (\omega, \hat{\omega}) \in \tilde{\Omega}. \label{varphi2}
	\end{equation}
Furthermore, for every $t \in [0,T]$ the function $\tilde{\mathcal{E}}_t(\tilde{X})$ is upper semianalytic and measurable with respect to $\mathcal{G}_t^{*}$ and $\mathcal{G}^{\mathcal{P}}$ and satisfies the following consistency condition, i.e., for every $t \in [0,T]$
	\begin{equation}
			\tilde{\mathcal{E}}_t (\tilde{X})=\esssupT_{\tilde{P}' \in \tilde{\mathcal{P}}(t;\tilde{P})} E^{\tilde{P}'}[\tilde{X}| \mathcal{G}_{t}] \quad \tilde{P}\text{-a.s. for all } \tilde{P} \in \tilde{\mathcal{P}}, \label{esssupBiagini}
	\end{equation}
where $\mathcal{\tilde{P}}(t;\tilde{P})=\lbrace \tilde{P}' \in \tilde{\mathcal{P}}: \tilde{P}'=\tilde{P} \text{ on } \mathcal{G}_{t}\rbrace$.
	\end{prop}
This family of sublinear conditional expectations $(\tilde{\mathcal{E}}_t)_{t \in [0,T]}$ is called $(\tilde{\mathcal{P}}, \mathbb{G})$-conditional expectation.
In general, the $(\tilde{\mathcal{P}}, \mathbb{G})$-conditional expectation does not satisfy a strong tower property as the $(\mathcal{P}, \mathbb{F})$-conditional expectation does in (\ref{towerNutz}). However, it is shown in Theorem 2.22 in \cite{bz_2019} that $(\tilde{\mathcal{E}}_t(\tilde{X}))_{t \in [0,T]}$ fulfills a weak form of time-consistency, which means
		\begin{equation*}
			\tilde{\mathcal{E}}_s(\tilde{\mathcal{E}}_t(\tilde{X})) \geq \tilde{\mathcal{E}}_s(\tilde{X}) \quad \text{for all } 0 \leq s \leq t \leq T \ \tilde{{P}} \text{-a.s. for all } \tilde{P} \in \tilde{\mathcal{P}}.
		\end{equation*}
Moreover, for the fundamental building blocks of life insurance liabilities, namely term insurance, annuity and pure endowment contract, it is proved in Proposition 2.31 in \cite{bz_2019} that the strong tower property is satisfied.
\begin{remark}
We note that the results in \cite{bz_2019} are also valid if we replace $\Omega$ by $C_x(\mathbb{R}_+,\mathbb{R}^d)$ or the space $D_x(\mathbb{R}_+,\mathbb{R}^d)$ of c\`{a}dl\`{a}g functions in $\mathbb{R}^d$ starting in $x \in \mathbb{R}^d$ equipped with the topology of locally uniform convergence or with the Skorokhod topology, respectively. In the last case a generalized version of Proposition \ref{SublinearNutz} is derived in Theorem 4.29 in \cite{hollender_phd}, where the definition of the concatenation process in $(\ref{concatenation})$ needs to be adapted in the following way
	\begin{equation*}
		(\omega \otimes_{\tau} \tilde{\omega})_t:= \omega_t \textbf{1}_{[0,\tau(\omega))}(t)+(\omega_{\tau(\omega)}+ \tilde{\omega}_{t-\tau(\omega)}-x)\textbf{1}_{[\tau(\omega)),\infty)}(t), \quad t \geq 0
	\end{equation*}
with $\omega, \tilde{\omega} \in \Omega_x$ and $\tau$ stopping time.
\end{remark}

\begin{remark}
In general, we do not need to assume the existence of the intensity process $\mu=(\mu_t)_{t \geq 0}$ as in \cite{bz_2019} in order to define $(\tilde{\mathcal{P}},\mathbb{G})$-conditional expectations, see $(\ref{definitionTau})$. However, with the representation of $\Gamma$ as in $(\ref{representationIntensity})$ we get more tractable results and $\mu$ will also be necessary for using the framework of \cite{fns_2019}. \\
The above construction of the product space $\tilde{\Omega}$ via the stopping time $\tilde{\tau}$ is a special case of the Cox model, see e.g., Remark 2.24 (a) in \cite{jeanblanc_aksamit}, which was suggested for modeling credit risk the first time in \cite{lando}. 
However, the construction of the stopping time $\tilde{\tau}$ in $(\ref{definitionTau})$ can be generalized to include other possible distributions for $\tilde{\tau}$. Note that this will change the definition of $(\tilde{\mathcal{E}}_t)_{t \in [0,T]}$ in $(\ref{BasicDecomposition})$. The Cox model can also be generalized by considering a process $\Gamma$ with  c\`{a}dl\`{a}g instead of continuous paths as done in \cite{jeanblanc_freiburg}. Unfortunately, we cannot transfer this case to the setting of \cite{bz_2019} as the construction of the operator $(\tilde{\mathcal{E}}_t)_{t \in [0,T]}$ requires the continuity of $\Gamma$, see the proofs of Lemma 2.10 and Proposition 2.11 in \cite{bz_2019}.	
\end{remark}

In the framework of insurance modeling we now wish to apply the above results to the valuation of insurance products under model uncertainty. In \cite{bz_2019} it is shown that the conditional sublinear operator $(\tilde{\mathcal{E}}_t)_{t \in [0,T]}$ can be used as pricing operator for life insurance liabilities. For simplicity we focus on endowments, i.e., contracts with payoff $\textbf{1}_{\lbrace \tilde{\tau} > T\rbrace }Y$, where $Y$ is an $\mathcal{F}_T^{*,\mathcal{P}}$-measurable nonnegative upper semianalytic function on $\Omega$ such that $\mathcal{E}(Y):=\sup_{P \in \mathcal{P}} E^P[Y] < \infty$. 
In this case the payment is made at the maturity of the contract only if the default event does not occur before the maturity date. For this contract the following valuation formula is deduced in Lemma 2.26 in \cite{bz_2019}.
\begin{lemma}
Let $Y=Y(\omega)$, $\omega \in \Omega$, be an $\mathcal{F}_T^{*,\mathcal{P}}$-measurable upper semianalytic function such that $\mathcal{E}(\vert Y \vert ) < \infty$. Then for every $t \in [0,T]$,
	\begin{equation*}
		\textbf{1}_{\lbrace \tilde{\tau} > T \rbrace } \quad \text{ and } \quad Ye^{-\int_t^T \mu_u du}
	\end{equation*}
are upper semianalytic functions and belong to $L^1(\tilde{\Omega})$.
Furthermore, if $(\mathcal{P}(t,\omega))_{(t,\omega) \in [0,T] \times \Omega}$ satisfies Assumption \ref{assumptionnutzNew}, the following holds pathwisely for every $t \in [0,T]$
	\begin{equation}\label{valuation}
		\tilde{\mathcal{E}}_t(Y\textbf{1}_{\lbrace \tilde{\tau} > T \rbrace }) = \textbf{1}_{\lbrace \tilde{\tau} > t \rbrace } \mathcal{E}_t(Ye^{-\int_t^T \mu_u du}).
	\end{equation}
\end{lemma}
In order to evaluate $(\ref{valuation})$ we wish to use the results on affine processes under parameter uncertainty from the paper \cite{fns_2019}. However, we need first to embed the framework of \cite{fns_2019} into our setting.  
\end{section}

\begin{section}{Affine processes under parameter uncertainty} \label{SectionAffineProcesses}
We now briefly recall the framework of \cite{fns_2019}. Consider a probability measure $P \in \mathcal{P}(\Omega)$ such that the canonical process $B$ is a continuous $(P,\mathbb{F})$-semimartingale such that $B=B_0 + M^P+ A^P$, where $A^P$ is a stochastic process with continuous paths of finite variation $P$-a.s., $M^P$ is a continuous $(P,\mathbb{F})$-local martingale and $A_0^P=M_0^P=0$. The $(P,\mathbb{F})$-characteristics of the semimartingale $B$ with such a decomposition are then given by the pair $(A^P,C)$ with $C=\langle M^P \rangle$. From now on we only consider semimartingales with absolutely continuous (a.c.) characteristics $(\beta^P,\alpha)$, i.e., with predictable processes $\beta^P$ and $\alpha \geq 0$ such that for all $t \in [0,T]$
\begin{equation*}
	A^P_t=\int_0^{t} \beta^P_s ds, \quad C_t=\int_0^{t} \alpha_s ds.
\end{equation*}
We define the set
\begin{equation*}
	\mathcal{P}_{sem}^{ac}=\lbrace P \in \mathcal{P}(\Omega) | B \text{ is a (} P,\mathbb{F} \text{)-semimartingale with a.c. characteristics} \rbrace.
\end{equation*}
To consider model risk a parameter vector $\theta=(b^0,b^1,a^0,a^1)$ is introduced. For $\underline{b}^i < \overline{b}^i, i=0,1,$ and $\underline{a}^i<\overline{a}^i, i=0,1,$ we define the compact set
\begin{equation} 
	\Theta: = \underbrace{[\underline{b}^0, \overline{b}^0] \times [\underline{b}^1, \overline{b}^1]}_{:=B} \times \underbrace{[\underline{a}^0, \overline{a}^0] \times [\underline{a}^1, \overline{a}^1]}_{:=A} \subset \mathbb{R}^2 \times \mathbb{R}^2_{\geq 0}. \label{Theta}
\end{equation}
Moreover, we define for $x \in \mathbb{R}$ the following set-valued functions for $A$ and $B$
\begin{align} \label{affineBounds}
	\begin{aligned}
	& b^*(x):=\lbrace b^0 + b^1 x: b \in B \rbrace,\\
	& a^*(x):=\lbrace a^0 + a^1 x^+: a \in A \rbrace,
	\end{aligned}
\end{align}
where $a:=(a^0,a^1),b:=(b^0,b^1) \in \mathbb{R}^2$, and $(\cdot)^+:=\max\lbrace \cdot, 0 \rbrace$. As $\Theta$ is an interval, $b^*(x)$ and $a^*(x)$ are intervals and can be described by
\begin{align*}
	& b^*(x)=[\underline{b}^0 + (\underline{b}^1 \textbf{1}_{\lbrace x \geq 0 \rbrace} + \overline{b}^1 \textbf{1}_{\lbrace x < 0 \rbrace})x,\overline{b}^0 + (\overline{b}^1 \textbf{1}_{\lbrace x \geq 0 \rbrace} + \underline{b}^1 \textbf{1}_{\lbrace x < 0 \rbrace})x ], \\
	&a^*(x)=[\underline{a}^0 + \underline{a}^1 x^+, \overline{a}^0+\overline{a}^1 x^+].
\end{align*}
\begin{defi} \cite[Definition 2.1]{fns_2019} \label{defiOrgiAffineDominated}
Let $\Theta$ be a set as in $(\ref{Theta})$ with associated $a^*, b^*$ as in $(\ref{affineBounds})$. Consider $t \in [0,T]$ and $P \in \mathcal{P}_{sem}^{ac}$ be a semimartingale law. We say that $P$ is \emph{affine-dominated on} $(t,T]$ \emph{by} $\Theta$, if $(\beta^P, \alpha)$ satisfy
	\begin{equation*}
		\beta_s^P \in b^*(B_s), \quad \alpha_s \in a^*(B_s),
	\end{equation*}
	for $dP \otimes dt$-almost all $(\omega,s) \in \Omega \times (t,T]$. If $t=0$, we say that $P$ is \emph{affine-dominated by} $\Theta$.
\end{defi}
Let $\mathcal{O}$ be the considered state space, i.e., either $\mathbb{R}$, $\mathbb{R}_{\geq 0}$ or $\mathbb{R}_{>0}$.
\begin{defi} \cite[Definition 2.2]{fns_2019}\label{defiAffineDominated}
Let $\Theta$ be a set as in $(\ref{Theta})$ with associated $a^*, b^*$ as in $(\ref{affineBounds})$. A family of semimartingale laws $P \in \mathcal{P}_{sem}^{ac}$ such that
\begin{enumerate}
\itemsep0pt
	\item $P(B_0=x)=1$,
	\item $P$ is affine dominated by $\Theta$
\end{enumerate}
is called \emph{affine process under parameter uncertainty starting at} $x \in \mathcal{O}$. 
If this holds for $P$, then we use the notation $P \in \mathcal{A}(x,\Theta)$.
Furthermore, for $t \in [0,T]$ we say that $P \in \mathcal{A}(t,x, \Theta)$ if $P \in \mathcal{P}_{sem}^{ac}$ and
\begin{enumerate}
\itemsep0pt
	\item $P(B_t=x)=1$,
	\item $P$ is affine dominated on $(t,T]$ by $\Theta$.
\end{enumerate}
\end{defi}
The state space  $\mathcal{O}$ and the parameter space $\Theta$ cannot be chosen to be completely independent. Otherwise, it can happen that the set $\mathcal{A}(x, \Theta)$ is empty. To avoid this problem, we introduce the definition of proper families of affine processes under parameter uncertainty.
\begin{defi} \label{properFamily}
	The families of non-linear affine process laws $(\mathcal{A}(x,\Theta))_{x \in \mathcal{O}}$ with state space $\mathcal{O}$ are called \emph{proper} if either $\underline{a}^0 >0$ or $\underline{a}^0=\overline{a}^0=0$ and $\underline{b}^0 \geq \frac{\overline{a}^1}{2}>0$ holds.
\end{defi}
\end{section}

\begin{section}{Reduced form setting under model uncertainty with non-linear affine intensities} \label{Section_Combination}
In the sequel, we include affine processes under parameter uncertainty in the setting of Section \ref{Section_Reduced_framework} to obtain analytically tractable models for credit risk/insurance markets under model uncertainty. 

\begin{subsection}{Market model on $\Omega_x$} \label{SubsectionMarketModel}
We consider the space $\Omega_x:=C_x([0,T],\mathbb{R}^2)$ of continuous functions with values in $\mathbb{R}^2$ starting at a fixed point $x \in \mathbb{R}^2$. An element of this space is denoted by $\omega:=(\omega^S, \omega^{\mu})$ and the canonical process $B:=(B^S, B^{\mu})$ given by $B_t(\omega) = (B_t^S(\omega), B_t^{\mu}(\omega))=(\omega^S_t, \omega^{\mu}_t)$, $t \in [0,T]$. We assume to be in the setting introducd in Section \ref{Section_Reduced_framework} applied to $\Omega_x$ and keep the notation as there. On $(\Omega_x, \mathcal{F})$ we consider a financial market model consisting of a riskfree asset $S^0 \equiv 1$ and of a risky asset $S=(S_t)_{t \in [0,T]}$ driven by $B^S$. The mortality intensity $\mu$ in $(\ref{representationIntensity})$ is given by $B^{\mu}$. As the mortality intensity is nonnegative, we assume that $B^{\mu}$ is nonnegative $\mathcal{Z}$-q.s. A sufficient condition that this assumption holds for a non-linear affine process defined as in Definition \ref{defiAffineDominated} is given in Proposition 2.3 in \cite{fns_2019}.

\begin{remark}
 Without loss of generality it is possible to consider a financial market consisting of $d$ risky assets by setting $\Omega_x:=C_x([0,T],\mathbb{R}^{d+1})$ with the canonical process $B_t(\omega):=(B^S_t(\omega),B^{\mu}_t(\omega)):=(\omega_{t}^{S^1},...,\omega_t^{S^d},\omega_t^{\mu})$ for $\omega \in \Omega^d_x, x \in \mathbb{R}^{d+1}, t \in [0,T]$ and $d \in \mathbb{N}$. In this case the $d$ assets $S^1,...,S^d$ are driven by the $d$-dimensional process $B^S$. 
\end{remark}

We define the families of probability measures $(\mathcal{Z}(t,\omega))_{(t,\omega) \in [0,T] \times \Omega_x}$ on $\Omega_x$  by
\begin{equation}
	\mathcal{Z}(t,\omega):=\mathcal{P}^S \cap \mathcal{A}^{\mu}(t,\omega_t^{\mu},\Theta^{\mu}) \label{definitionSets}
\end{equation}
for $(t,\omega) \in [0,T] \times \Omega_x$, where the set $\mathcal{A}^{\mu}(t,\omega_t^{\mu},\Theta^{\mu})$ is introduced next.

\begin{defi}\ \label{IntensityP}
Let $(t,\overline{\omega}) \in [0,T] \times \Omega_x$, $\Theta^{\mu}$ as in $(\ref{Theta})$,  $b^*(B_{\cdot}^{\mu})$ and $a^*(B_{\cdot}^{\mu})$ as in $(\ref{affineBounds})$. Given $P^{\mu} \in \mathcal{P}(\Omega_x)$  we have that $P^{\mu} \in \mathcal{A}^{\mu}(t,\overline{\omega}_t^{\mu},\Theta^{\mu})$ if
\begin{enumerate}
\itemsep0pt
	\item the process $B^{\mu}$ is a one-dimensional $(P^{\mu},\mathbb{F})$-semimartingale with a.c. characteristics with the corresponding predictable processes $\beta^{P^{\mu}}$ and $\alpha \geq 0$,
	\item $P^{\mu}(B_t^{\mu}=\overline{\omega}_t^{\mu})=P^{\mu}(\lbrace \omega \in \Omega_x: B_t^{\mu}(\omega)=\overline{\omega}^{\mu}_t \rbrace)=1$,
	\item $P^{\mu}$ is \emph{affine-dominated by} $(t,T]$ by $\Theta^{\mu}$, i.e., $\beta^{P^{\mu}}_s \in b^*(B_s^{\mu})$ and $\alpha_s \in a^*(B_s^{\mu})$ for $dP^{\mu} \otimes dt$-almost all $(\omega, s) \in \Omega_x \times (t, T]$.
\end{enumerate}

\end{defi}
In addition, the following assumptions hold for the set $\mathcal{P}^S \subseteq \mathcal{P}(\Omega_x)$ in $(\ref{definitionSets})$.
\begin{asum}\ \label{AssumptionMartingale}
\begin{enumerate}
\itemsep0pt
	\item $\mathcal{P}^S$ satisfies Assumption \ref{assumptionnutzNew}, where we define the families $(\mathcal{P}^S(t,\omega))_{(t,\omega) \in [0,T] \times \Omega_x}$ by $\mathcal{P}^S(t,\omega):=\mathcal{P}^S$ for all $(t,\omega) \in [0,T] \times \Omega_x$.
	\item For all $P \in \mathcal{P}^S$ the process $S=(S_t)_{t \in [0,T]}$ is a $(P,\mathbb{F})$-local martingale.
	\end{enumerate}
\end{asum}
Note that by definition the set $\mathcal{Z}(0,\omega)$ is independent of $\omega$, which is crucial for Proposition \ref{SublinearNutz}. From now on, set $\mathcal{Z}:=\mathcal{Z}(0,\omega)$.   

\begin{remark} 
By the definition of the families of probability measures $(\mathcal{Z}(t,\omega))_{(t,\omega) \in [0,T] \times \Omega_x}$ in $(\ref{definitionSets})$ 
the mortality intensity $B^{\mu}$ is a non-linear affine process and the risky asset $S$ is modeled in an arbitrage-free way under model uncertainty. In a more general setting one could work with a set $\mathcal{P}^{S}$ of semimartingale measures for $S$ and follow the approach of \cite{bbkn_2017}. However, in this case the existence of an equivalent local martingale measure can only be guaranteed by considering an additional cemetery state at which the paths jump at the stopping time $\xi$. To avoid further technicalities and to be consistent with classical reduced form models, we directly assume that $\mathcal{P}^S$ are local martingale measures for $S$. Moreover, for our purpose it is sufficient to consider a set of probability measures $\mathcal{P}^S$ instead of the families $(\mathcal{P}^S(t,\omega))_{(t,\omega) \in [0,T] \times \Omega_x}$. However, our results can be easily extended for families by requiring that the local martingale property holds for all $P \in \mathcal{P}^S(0,\omega)$.
\end{remark}

We now show that Assumption \ref{assumptionnutzNew} holds for $(\mathcal{Z}(t,\omega))_{(t,\omega) \in [0,T] \times \Omega_x}$, so that we can define the sublinear operator $(\mathcal{E}_t)_{t \in [0,T]}$ on $\Omega_x$ as in Proposition \ref{SublinearNutz} with respect to $ (\mathcal{Z}(t,\omega))_{(t,\omega) \in [0,T] \times \Omega_x}$ in $(\ref{definitionSets})$.
\begin{prop} \label{intersectionAssumption}
	Let $(\mathcal{Q}(t,{\omega}))_{(t,\omega) \in [0,T] \times \Omega_x}$ and   $(\tilde{\mathcal{Q}}(t,{\omega}))_{(t,\omega) \in [0,T] \times \Omega_x}$ be two families satisfying Assumption \ref{assumptionnutzNew}. Then the families $(\mathcal{P}(t,{\omega}))_{(t,\omega) \in [0,T] \times \Omega_x}$ defined by
	\begin{equation*}
		\mathcal{P}(t,{\omega}):= \mathcal{Q}(t,{\omega}) \cap \mathcal{\tilde{Q}}(t,{\omega}), \quad (t,\omega) \in [0,T] \times \Omega_x
	\end{equation*}
also fulfill Assumption \ref{assumptionnutzNew}.
\end{prop}
\begin{proof}
Let $(s, \overline{\omega}) \in [0,T] \times \Omega_x$, $P \in \mathcal{P}(s,\overline{\omega})$ and $\tau$ be a stopping times such that $ s \leq \tau \leq T$. Set $\eta:=\tau^{s,\overline{\omega}}-s$.\\
	1) \emph{Measurability:} As $\mathcal{Q}(s,\overline{\omega})$ and $\mathcal{\tilde{Q}}(s,\overline{\omega})$ satisfy Assumption \ref{assumptionnutzNew}, we know that the sets
		$\lbrace (P',\omega): \omega \in \Omega, P' \in \mathcal{Q}(\tau,{\omega})\rbrace$ and 
		$\lbrace (P',\omega): \omega \in \Omega, P' \in \mathcal{\tilde{Q}}(\tau,{\omega})\rbrace$
	are analytic. As the countable intersection of analytic sets is again analytic, the property of measurability also holds for $(\mathcal{P}(t,{\omega}))_{(t,\omega) \in [0,T] \times \Omega_x}$.\\
	2) \emph{Invariance:} As $P \in \mathcal{Q}(s,\overline{\omega}) \cap \mathcal{\tilde{Q}}(s,\overline{\omega})$ we get
	\begin{align*}
		&P^{\eta,\omega} \in \mathcal{Q}(\tau,\overline{\omega} \otimes_s \omega) \text{ for } P\text{-a.e. } \omega \in \Omega_x \text{ and}\\
		&P^{\eta,\omega} \in \mathcal{\tilde{Q}}(\tau,\overline{\omega} \otimes_s \omega) \text{ for } P\text{-a.e. } \omega \in \Omega_x
	\end{align*}
	and we can conclude  
	\begin{equation*}
		P^{\eta,\omega} \in \underbrace{\mathcal{Q}(\tau,\overline{\omega} \otimes_s \omega) \cap \mathcal{\tilde{Q}}(\tau,\overline{\omega} \otimes_s \omega)}_{=\mathcal{P}(\tau,\overline{\omega} \otimes_s \omega) }\text{ for } P\text{-a.e. } \omega \in \Omega_x.
	\end{equation*}
	3) \emph{Stability under Pasting:} Let $\kappa: \Omega_x \to \mathcal{P}(\Omega_x)$ be an $\mathcal{F}_{\eta}$-measurable kernel and $\kappa(\omega) \in \mathcal{P}(\tau, \overline{\omega} \otimes_s \omega)$ for $P$-a.e. $\omega \in \Omega_x$. Due to the fact that $P \in \mathcal{Q}(s,\overline{\omega}) \cap \mathcal{\tilde{Q}}(s,\overline{\omega})$, we get for $A \in \mathcal{F}$
	\begin{align*}
		&\overline{P}(A)= \int \int (\textbf{1}_A)^{\eta,\omega}(\omega') \kappa (d \omega', \omega) P(d \omega) \in \mathcal{Q}(s,\overline{\omega}) \text{ and}\\
		&\overline{P}(A)= \int \int (\textbf{1}_A)^{\eta,\omega}(\omega') \kappa (d \omega', \omega) P(d \omega) \in \mathcal{\tilde{Q}}(s, \overline{\omega}). 
	\end{align*} 
So it follows $\overline{P}(A) \in \mathcal{P}(s,\overline{\omega})$. 
\end{proof}
Next we show that $(\mathcal{A}^{\mu}(t,\omega_t^{\mu},\Theta^{\mu}))_{(t,\omega) \in [0,T] \times \Omega_x}$ satisfies Assumption \ref{assumptionnutzNew}. To do so, we first use the results of Lemma 3.1 and 3.2 in \cite{fns_2019} for the set $\mathcal{A}(t,\omega_t,\Theta)$ with $(t,\omega) \in [0,T] \times \Omega_x^1$ defined in Definition \ref{defiAffineDominated}. Note, here we work with the one-dimensional path space $\Omega_x^1=C_x([0,T],\mathbb{R})$ for a fixed $x \in \mathbb{R}$. Then we prove that the families $(\mathcal{A}^{\mu}(t,\omega_t^{\mu},\Theta^{\mu}))_{(t,\omega) \in [0,T] \times \Omega_x}$ in Definition \ref{IntensityP} satisfy Assumptions \ref{assumptionnutzNew}. This requires to transfer the results in \cite{nn_levy_2016} to a two-dimensional setting.

\begin{prop} \label{crucialPoint}
	Let $(\mathcal{A}^{\mu}(s,\omega^{\mu}_s, \Theta^{\mu}))_{(s,\omega) \in [0,T] \times \Omega_x}$ be the families as in Definition \ref{IntensityP} for a fixed set $\Theta^{\mu}$ as in $(\ref{Theta})$. 
Fix $s \in [0,T], \overline{\omega} \in \Omega_x$ and $\tau$ be a stopping time taking values in $[s,T]$. Moreover, let $P \in \mathcal{A}^{\mu}(s,\overline{\omega}^{\mu}(s), \Theta^{\mu})$, then
	\begin{enumerate}
	\itemsep0pt
		\item \emph{Measurability:} The graph $\lbrace (P',\omega): \omega \in \Omega_x, P' \in \mathcal{A}^{\mu}(\tau(\omega),\omega^{\mu}_{\tau(\omega)}, \Theta^{\mu}) \rbrace \subseteq \mathcal{P}(\Omega_x) \times \Omega_x$ is analytic.
		\item \emph{Invariance:} There exists a family of conditional probabilities $(\mathcal{P}^{\tau,\omega})_{\omega \in \Omega_x}$ with respect to $\mathcal{F}_{\tau}$ such that $P^{\tau,\omega} \in \mathcal{A}^{\mu}(\tau(\omega),\omega_{\tau(\omega)}^{\mu}, \Theta^{\mu})$ for $P$-a.e. $\omega \in \Omega_x$.
	 	\item \emph{Stability under Pasting:} Assume that there exists a family of probability measures $(Q^{\omega})_{\omega \in \Omega_x}$ such that $Q^{\omega} \in \mathcal{A}^{\mu}(\tau(\omega),\omega_{\tau(\omega)}^{\mu}, \Theta^{\mu})$ for $P$-a.e. $\omega \in \Omega_x$ and the map $\omega \to Q^{\omega}$ is $\mathcal{F}_{\tau}$ measurable. Then the probability measure $P \otimes Q$ defined by
	  	\begin{equation*}
	 		P \otimes Q (\cdot)=\int_{\Omega_x} Q^{\omega}(\cdot) P(d\omega)
	 	\end{equation*}
	 	is an element of $\mathcal{A}^{\mu}(s,\overline{\omega}^{\mu}(s),\Theta^{\mu})$.
	\end{enumerate} 
\end{prop}

\begin{proof}
See Appendix.
\end{proof}
The statements in Proposition \ref{crucialPoint} do not correspond directly to Assumption \ref{assumptionnutzNew} but are in line with the ones used in \cite{el_tan_2013}, \cite{el_tan_2015}. In these papers families of probability measures $(\mathcal{P}(t,\omega))_{(t,\omega) \in [0,T] \times \Omega}$ with support on $\mathcal{D}_{(\omega,t)}:=\lbrace \tilde{\omega}:  \tilde{\omega} \vert_{[0,t]} = \omega \vert_{[0,s]} \rbrace$ are considered. As a consequence the concatenation is defined as 
	\begin{equation}
		(\omega \otimes_t \tilde{\omega})_s:=\omega_s \textbf{1}_{\lbrace s \leq t \rbrace} + \tilde{\omega}_s \textbf{1}_{\lbrace s \geq t \rbrace} \label{concatenationElKaroui}
	\end{equation}
	for $\omega,\tilde{\omega} \in \Omega$, such that $\omega_t=\tilde{\omega}_t$. Note, this is equivalent to $(\ref{concatenation})$ if $\omega_t=\tilde{\omega}_t$, i.e., $(\ref{concatenation})$ extends $(\ref{concatenationElKaroui})$ to the case where $\omega_t \neq \tilde{\omega}_t$. Therefore, the equivalence between Assumption \ref{assumptionnutzNew} and the formulation in  Proposition \ref{crucialPoint} is clear by taking into account the two different conventions in $(\ref{concatenation})$ and $(\ref{concatenationElKaroui})$, respectively.

\begin{remark}
	In \cite{fns_2019} it is stated that the invariance and stability under pasting properties of the families $(\mathcal{A}(t,\omega_t,\Theta))_{(t,\omega) \in [0,T] \times \Omega_x^1}$ in Lemma 3.1 and 3.2 in \cite{fns_2019} follow directly by Theorem 2.1 in \cite{nn_levy_2016}. The latter paper deals with a special case of our setting, as the differential characteristics are assumed to be in a fixed set $\Theta \subset \mathbb{R} \times \mathbb{S}^d_+ \times \mathcal{L}$, where $\mathbb{S}^d_+$ is the family of all positive definite symmetric real-valued $(d \times d)$-matrices and $\mathcal{L}$ the set of L\'evy measures. It is shown in Theorem 2.1 in \cite{nn_levy_2016} that 
\begin{equation}
	\mathcal{P}_{\Theta}=\lbrace P \in \mathcal{P}_{sem}^{ac}: (\beta^P, \alpha^P, F^P) \in \Theta, P\otimes dt \text{-a.e.}\rbrace, \label{ThetaLevy}
\end{equation}
where $\emptyset \neq 
	 \Theta \subseteq \mathbb{R}^d \times \mathbb{S}^d_+ \times \mathcal{L}$ is a Borel measurable subset, satisfies Assumption \ref{assumptionnutzNew}. The main difference with respect to our setting is that $\Theta$ is fixed in \cite{nn_levy_2016} and does not depend on the state of the canonical process $B,$ as in the affine case. In addition, \cite{nn_levy_2016} takes only into account a family $\mathcal{P}$ independent of $(s,\omega) \in \mathbb{R}_+ \times \Omega$. So, we cannot directly conclude that Assumption \ref{assumptionnutzNew} holds for the families $(\mathcal{A}(t,\omega_t,\Theta))_{(t,\omega) \in [0,T] \times \Omega_x^1}$ in our case. 
\end{remark}

In the following we present some examples for the families $(\mathcal{Z}(t,\omega))_{(t,\omega) \in [0,T] \times \Omega_x}$ in $(\ref{definitionSets})$.
\begin{example} \label{Example2RandomG-expectation}
The following sets are considered
	\begin{align*}
		\mathcal{M}&= \lbrace P \in \mathcal{P}(\Omega_x): B \text{ is a local } P \text{-martingale} \rbrace \\
		\mathcal{M}_a&= \lbrace P \in \mathcal{M}: \langle B \rangle^P \text{ is absolutely continuous } P \text{-a.s.} \rbrace,
	\end{align*}
where $B$ is the canonical process and $\langle B \rangle^P$ is the $\mathbb{R}^{2 \times 2}$-valued quadratic variation process of $B$ under $P$. For the sake of simplicity, we here assume $S=B^S$. 
Now, we define a set-valued process $\textbf{D}: \Omega_x \times [0,T] \to 2^{\mathbb{R}^{d \times d}}$ as in \cite{nutz_G-expectation} and consider the following families of probability measures $\mathcal{P}_{\textbf{D}}(s, \overline{\omega})$ for all $(s, \overline{\omega}) \in [0,T] \times \Omega_x$ given by
	\begin{equation}
	\mathcal{P}_{\textbf{D}}(s,\overline{\omega}):= \bigg \lbrace P \in \mathcal{M}_a: \frac{d \langle B \rangle_u^P}{du} (\omega) \in \textbf{D}_{u+s}(\overline{\omega} \otimes_s \omega) \text{ for } du \times P \text{-a.e. } (u,\omega) \in [0,T] \times \Omega_x \bigg \rbrace \label{Set_random_G}
	\end{equation}
	Thereby, it is assumed that the set $\textbf{D}$ satisfies the following regularity assumption.
	\begin{asum} \label{AssumptionG-Expectation}
	For every $t \in [0,T],$
		\begin{equation*}
			\lbrace (s,\omega,M) \in [0,t] \times \Omega_x \times \mathbb{R}^{d \times d}: M\in \textbf{D}_s(\omega) \rbrace \in \mathcal{B}([0,t]) \otimes \mathcal{F}_t \otimes \mathcal{B}(\mathbb{R}^{d \times d}).
		\end{equation*}
	\end{asum}
	Using this condition it is shown in Theorem 4.3 in \cite{nh_2013} that $\mathcal{P}_{\textbf{D}}(\tau,\overline{\omega})$ satisfy Assumption \ref{assumptionnutzNew}, where $\tau$ is a (finite) stopping time and $\overline{\omega} \in \Omega$. \\
	For every $(s,\overline{\omega}) \in [0,T] \times \Omega_x$ set $\mathcal{Z}(s,\overline{\omega}) := P_{\textbf{D}}(s,\overline{\omega})$ and define the process $\textbf{D}=(\textbf{D}_t)_{t \in [0,T]}$ with values in a subset of $\mathbb{R}^{d \times d}$ such that
	\begin{equation*}
		\textbf{D}_t(\omega)=\textbf{D}_t(\omega^{S},\omega^{\mu}) \in \bigg \lbrace \left(
  			 \begin{array}{ccc}
   			 \sigma^S & 0 \\
     		0  & a^{\mu,0} + a^{\mu,1} \omega_t^{\mu}
    		\end{array}
			\right): \sigma^S \in [\underline{\sigma}^S, \overline{\sigma}^S], a^{\mu,0} \in [\underline{a}^0, \overline{a}^0], a^{\mu,1} \in [\underline{a}^1, \overline{a}^1] \bigg \rbrace.
	\end{equation*}
	We first verify that the set $\mathcal{P}_{\textbf{D}}(s,\overline{\omega})$ in $(\ref{Set_random_G})$ with our choice of the process $\textbf{D}$ contains the probability measures we are interested in. To do so, fix $(s,\overline{\omega}) \in [0,T] \times \Omega_x$ and consider $P \in \mathcal{P}_{\textbf{D}}(s,\overline{\omega})$. Then, for $du \times P$-a.e. $(u, \omega) \in [0,T] \times \Omega_x$ we have  
	\begin{equation*}
		\frac{d \langle B \rangle_u^P}{du} (\omega) \in \textbf{D}_{u+s}(\overline{\omega} \otimes_s \omega) \text{ for } \sigma^S \in [\underline{\sigma}^S, \overline{\sigma}^S],  
		a^{\mu,0} \in [\underline{a}^0, \overline{a}^0], a^{\mu,1} \in [\underline{a}^1, \overline{a}^1]  
	\end{equation*}
	if and only if
	\begin{align}
		&\frac{d \langle B^{S}, B^{\mu}\rangle_u^P}{du} (\omega) = 0, \quad \frac{d \langle B^{S} \rangle_u^P}{du} (\omega) \in [\underline{\sigma}^S, \overline{\sigma}^S] , \quad \nonumber \\ 
		&\alpha_u^{P,\mu}(\omega) =\frac{d \langle B^{\mu} \rangle_u^P}{du} (\omega) \in  [\underline{a}^0 + \underline{a}^1 (\overline{\omega^{\mu}} \otimes_s \omega^{\mu})_{u+s} , \overline{a}^0 + \overline{a}^1 (\overline{\omega^{\mu}} \otimes_s \omega^{\mu})_{u+s}] = a^*((\overline{\omega^{\mu}} \otimes_s \omega^{\mu})_{u+s}) \label{Example_G-Expectation}.
	\end{align}
	Note, the statement in $(\ref{Example_G-Expectation})$ is equivalent to 
	\begin{align*}
		& \frac{d \langle B^{S}, B^{\mu}\rangle_u^P}{du} (\omega) = 0, \quad \frac{d \langle B^{S} \rangle_u^{\overline{P}}}{du} (\omega) \in [\underline{\sigma}^S, \overline{\sigma}^S] , \quad  \\ 
		&\alpha_u^{\overline{P},\mu}(\omega)  \in a^*( \omega^{\mu}_u) \text{ for } du \times \overline{P} \text{-a.e. } (u,\omega) \in [s,T] \times \Omega_x, \overline{P} \in \mathcal{P}_{\textbf{D}}(s,\overline{\omega}), \overline{P}(B_s^{\mu}= \overline{\omega}^{\mu}_s)=1,
	\end{align*}
	for $du \times P$-a.e. $(u, \omega) \in [0,T] \times \Omega_x$ with $P \in \mathcal{P}_{\textbf{D}}(s,\overline{\omega})$,
	i.e., $P$ is affine-dominated in the sense of Definition \ref{IntensityP}. 
	Assumption \ref{AssumptionG-Expectation} is obviously satisfied as the set-valued function $\textbf{D}$ is defined by elementary operations. In the outlined setting $B^{\mu}$ is an affine process under parameter uncertainty with $\Theta=\lbrace 0 \rbrace \times \lbrace 0 \rbrace \times [\underline{a}^0, \overline{a}^0] \times [\underline{a}^1, \overline{a}^1]$ and the process $B^S$ is a local martingale for all $P \in \mathcal{P}_{\textbf{D}}(s,\overline{\omega})$.
	In this case the set $\mathcal{Z}$ is saturated. \\
By defining the set-valued process $\textbf{D}$ as a constant process with values in a nonempty, convex and compact set of matrices $\overline{\textbf{D}} \subseteq \mathbb{R}^{2 \times 2}$, we are in the case of the classical $G$-setting. The definition of the families of probability measures $(\mathcal{P}_{\textbf{D}}(s,\overline{\omega}))_{(s,\overline{\omega}) \in [0,T] \times \Omega_x} $ in $(\ref{Set_random_G})$ reduces to
\begin{equation*}
		\mathcal{P}_{\textbf{D}} = \lbrace P \in \mathcal{M}_a: d \langle B \rangle_t^P/dt \in \overline{\textbf{D}} \ P \times dt\text{-a.e.} \rbrace
	\end{equation*}
for all $(s,\overline{\omega}) \in [0,T] \times \Omega_x$. It is shown in Proposition 3.1 in \cite{nh_2013} that $\mathcal{P}_{\textbf{D}}$ satisfies Assumption \ref{assumptionnutzNew}. Assume $\overline{\textbf{D}}$ contains matrices of the form $\tilde{D}=(\tilde{d}_{i,j})_{1 \leq i,j \leq 2}$ with $\tilde{d}_{1,2}=\tilde{d}_{2,1}=0$, $\tilde{d}_{1,1}= \sigma^S$ and $\tilde{d}_{2,2}= \sigma^{\mu}$ with $\sigma^S \in [\underline{\sigma}^S, \overline{\sigma}^S]$ and $\sigma^{\mu} \in [\underline{\sigma}^{\mu}, \overline{\sigma}^{\mu}]$ for $0 < \underline{\sigma}^S \leq  \overline{\sigma}^{S}$, $0 < \underline{\sigma}^{\mu} \leq  \overline{\sigma}^{\mu}$. Then this corresponds to the affine structure of the semimartingale components of $B^{\mu}$ with $\Theta^{\mu}= \lbrace 0 \rbrace \times \lbrace 0 \rbrace \times [\underline{\sigma}^{\mu}, \overline{\sigma}^{\mu}] \times \lbrace 0\rbrace$ and volatility uncertainty for $B^S$. Note, $\overline{\textbf{D}}$ is convex and also compact as the set of diagonalizable matrices with bounded eigenvalues is compact.
\end{example}
	
\begin{example} \label{Example3}
The last example can be generalized by considering an affine structure on the drift of $B^{\mu}$, i.e., $\Theta=[\underline{b}^0, \overline{b}^0] \times [\underline{b}^1, \overline{b}^1] \times [\underline{a}^0, \overline{a}^0] \times [\underline{a}^1, \overline{a}^1]$. Therefore, define the set-valued process $\textbf{L}: \Omega_x \times [0,T] \to 2^{\mathbb{R}^2}$ by
\begin{equation*}
	\textbf{L}_t(\omega^{S},\omega^{\mu}) \in \bigg \lbrace \left(
  		 \begin{array}{ccc}
   		 0 \\
    	 b^{\mu,0} + b^{\mu,1} \omega_t^{\mu}
    	\end{array}
		\right): b^{\mu,0} \in [\underline{b}^0, \overline{b}^0], b^{\mu,1} \in [\underline{b}^1, \overline{b}^1] \bigg \rbrace.
\end{equation*}
For all $(s, \overline{\omega}) \in [0,T] \times \Omega_x$ we set
\begin{align}
	\mathcal{Z}(s,\overline{\omega}): =\bigg \lbrace P \in \mathcal{P}_{sem}^{ac}: \frac{d [ B ]_u^P}{du} (\omega) \in \textbf{D}_{u+s} (\overline{\omega} \otimes_s \omega), \quad \frac{dA^P_u}{du} (\omega) \in \textbf{L}_{u+s} (\overline{\omega} \otimes_s \omega) \\
	\text{ for } du \times P \text{-a.e. } (u,\omega) \in [0,T] \times \Omega_x \bigg \rbrace,  \label{extendedSet}
\end{align}
where $A^P$ denotes the finite variation part of the semimartingale decomposition of $B$. In Proposition 4.3 in \cite{hollender_phd} it is shown that Assumption \ref{assumptionnutzNew} is satisfied under the following condition.
\begin{asum} \label{asumHollender}
For every $t \in [0,T]$
\begin{equation*}
	\lbrace (s,\omega,M,N) \in [0,t] \times \Omega_x \times \mathbb{R}^{2 \times 2} \times \mathbb{R}^{2} : M \in \textbf{D}_s(\omega), N \in \textbf{L}_s(\omega) \rbrace \in \mathcal{B}([0,t]) \otimes \mathcal{F}_t \otimes \mathcal{B}(\mathbb{R}^{2\times 2} \times \mathbb{R}^2).
\end{equation*}
\end{asum}
	In this example the local martingale property of $B^S$ is satisfied on $[t,T]$ for all $P \in \mathcal{Z}(t,\omega)$. 
\end{example}
\end{subsection}

\begin{subsection}{Extended market model on $\Tilde{\Omega}_x$} \label{SubsectionExtendedMarketModel}
Given the market model on $(\Omega_x, \mathcal{F}, \mathbb{F})$, a final time horizon $T>0$ and the families of probability measures $(\mathcal{Z}(t,\omega))_{(t,\omega) \in [0,T] \times \Omega_x}$ in $(\ref{definitionSets})$, we now use the construction described in Section $\ref{Section_Reduced_framework}$ to define an extended market model on $(\tilde{\Omega}_x,\mathcal{G}, \mathbb{G})$. As in $(\ref{probExtendedDependence})$ the families of probability measures $(\tilde{\mathcal{Z}}(t,\omega))_{(t,\omega) \in [0,T] \times \Omega_x}$ are given by
\begin{equation*}
	\tilde{\mathcal{Z}}(t,\omega):=\lbrace \tilde{P} \in \mathcal{P}(\tilde{\Omega}_x): \tilde{P} = P \otimes \hat{P}, P \in \mathcal{Z}(t,\omega) \rbrace 
\end{equation*}
for $(t,\omega) \in [0,T] \times \Omega_x$ with $\tilde{\mathcal{Z}}:=\tilde{\mathcal{Z}}(0,\omega)$. \\
Note, by Assumption \ref{AssumptionMartingale} for all $(t,\omega) \in [0,T] \times \Omega_x$ and $P \in \mathcal{Z}(t,\omega)$ the process $S$ is a local $(P, \mathbb{F})$-martingale. Let $\tilde{P} \in \tilde{\mathcal{Z}}(t,\omega)$, $0 \leq s \leq t \leq T$ and $(\tau_n)_{n \in \mathbb{N}}$ a suitable sequence of stopping times, then
	\begin{equation*}
		E^{\tilde{P}}[S_{t \wedge \tau_n} \vert \mathcal{F}_s] = E^{P \otimes \hat{P}}[S_{t \wedge \tau_n} \vert \mathcal{F}_s]= E^{P}[S_{t \wedge \tau_n} \vert \mathcal{F}_s] =S_{s \wedge \tau_n},
	\end{equation*}
which means that $S$ is also a $(\tilde{P},\mathbb{F})$-local martingale for every $\tilde{P} \in \tilde{\mathcal{Z}}(t,\omega)$. By construction (see Section 6.5 in \cite{bielecki_rutkowski_2004}) we have that in our setting the immersion property holds, i.e., every $\mathbb{F}$-martingale is also a $\mathbb{G}$-martingale. Note that the usual hypotheses on the filtrations are not necessary for the immersion property to hold. Thus, $S$ is also a local $(\tilde{P},\mathbb{G})$-martingale for every $\tilde{P} \in \tilde{\mathcal{Z}}(t,\omega)$, $(t,\omega) \in [0,T] \times \Omega_x$. Assumption \ref{AssumptionMartingale} implies that $S$ is also a local martingale for all $\tilde{P} \in \tilde{\mathcal{Z}}(t, \omega)$ with $(t,\omega) \in [0,T] \times \Omega_x$. 
\end{subsection}
\end{section}

\begin{section}{Longevity bond under model uncertainty} \label{SectionMainLongevity}
In the following we introduce a longevity bond with price process $S^L:=(S^L_t)_{t \in [0,T]}$ and maturity $T>0$ by means of the survivor index as in Section 2.1.2 in \cite{cbd_2006} on the financial market $(\tilde{\Omega}_x,\mathcal{G})$. 
	The survivor index $S^{\text{sur}}=(S_t^{\text{sur}})_{t \in [0,T]}$ is defined by
		\begin{equation} \label{survivor}
		S^{\text{sur}}_t=\exp \bigg( - \int_0^t B^{\mu}_s ds \bigg), \quad t \in [0,T],
		\end{equation}	
	where $B^{\mu}=(B^{\mu}_s)_{s \in [0,T]}$ represents the mortality intensity of a fixed age cohort. 
\begin{defi} \label{defiLongevityClassical}
	A longevity bond with maturity $T$ is a bond paying the amount $S^{\text{sur}}_T$ at time $T$. 	
\end{defi}
In the sequel, our aim is to introduce the price process $S^L:=(S_t^L)_{t \in [0,T]}$ for a longevity bond with maturity $T$ under model uncertainty in a way that the resulting extended market $(S^0,S,S^L)$ is arbitrage-free in the sense of Definition \ref{DefiNA1ExtendedMarket}.\\
In \cite{fns_2019} the upper bond prices under the non-linear affine term structure model $\mathcal{A}(t,x,\Theta), x \in \mathcal{O}$ is defined as \begin{equation} 
\overline{p}(t,T,x):=\sup_{P \in \mathcal{A}(t,x,\Theta)} E^P[e^{-\int_t^T B^{\mu}_s ds}  |B_t^{\mu} = x]	, \quad 0 \leq t \leq T. \label{upperBond}
\end{equation}
This bond price is then given as the solution of generalized Riccati equations in Proposition 6.2 in \cite{fns_2019} and in some important special case leads to a closed-form solution.
Note, that 
	\begin{align*}
		\overline{p}(t,T,x)=&\sup_{P \in \mathcal{A}(t,x,\Theta)} E^P[e^{-\int_t^T B_s^{\mu} ds}  | B_t^{\mu} = x]  
			=\sup_{P \in \mathcal{A}(t,x,\Theta)}  E^P[e^{-\int_t^T B_s^{\mu} ds}  \textbf{1}_{\lbrace B_t^{\mu} = x \rbrace}],
	\end{align*}
where we have used that $P(B_t^{\mu}=x)=1$ for all $P  \in \mathcal{A}(t,x,\Theta)$. Hence $\overline{p}(t,T,x)$ represents the worst case estimation for the bond price given the class of models $\mathcal{A}(t,x,\Theta)$. However, it will not in general coincide with the superreplication price for the defaultable contingent claim $H=\textbf{1}_{\lbrace \tau > t \rbrace}$. We now discuss a more general definition for the longevity bond in our setting. \\
The first candidate for $S^L$ is the process $Y=(Y_t)_{t \in [0,T]}$ given by
\begin{equation}
	Y_t:= \tilde{\mathcal{E}}_t(e^{-\int_0^T B_s^{\mu} ds })= \esssupT_{\tilde{P}' \in \tilde{\mathcal{Z}}(t;P)} E^{\tilde{P}'} [e^{-\int_0^T B_s^{\mu} ds} \vert \mathcal{G}_t ] \quad \tilde{P} \text{-a.s. for all } \tilde{P} \in \tilde{\mathcal{Z}},\label{LongevityEsssupExtended}
\end{equation}
where  $\tilde{\mathcal{Z}}(t;\tilde{P}):=\lbrace \tilde{P}' \in \tilde{\mathcal{Z}}: \tilde{P} = \tilde{P}' \text{ on } \mathcal{G}_t \rbrace$ and $(\tilde{\mathcal{E}}_t)_{t \in [0,T]}$ is the conditional sublinear operator introduced in $(\ref{BasicDecomposition})$. 
By Proposition \ref{extendedOperator} the process $Y$ is $\mathbb{G}^*$-adapted and well-defined as $e^{-\int_0^T B_s^{\mu} ds }$ is a nonnegative Borel-measurable function. Moreover, $\tilde{\mathcal{E}}_t(e^{-\int_0^T B_s^{\mu} ds })$ coincides with $\mathcal{E}_t(e^{-\int_0^T B_s^{\mu} ds })$ by Remark 2.19 1. in \cite{bz_2019}, which means $Y$ is $\mathbb{F}^*$-measurable by Proposition \ref{SublinearNutz}. \\
Unfortunately, the process $Y$ has no  c\`{a}dl\`{a}g paths which is often necessary for standard results in financial mathematics. Motivated by the proof of Theorem 3.2 in \cite{n_2015} we define the value of the longevity bond as c\`{a}dl\`{a}g process as follows.
\begin{defi} \label{defiLongevityBond}
The value process of the longevity bond $S^L:=(S_t^L)_{t \in [0,T]}$ is given by $S^L:=Y' \textbf{1}_{N^c}$, where $N$ belongs to the family $\mathcal{N}^{\tilde{\mathcal{Z}}}_T$ of $(\tilde{P},\mathcal{G}_T)$-null sets for all $\tilde{P} \in \tilde{\mathcal{Z}}$ denoted by $\mathcal{N}^{\tilde{\mathcal{Z}}}_T$ and $Y'$ is given by
\begin{align}
	Y_t' &:= \limsup_{r \downarrow t, r \in \mathbb{Q}} Y_r=\limsup_{r \downarrow t, r \in \mathbb{Q}} \tilde{\mathcal{E}}_r(e^{-\int_0^T B_s^{\mu} ds })=\limsup_{r \downarrow t, r \in \mathbb{Q}} {\mathcal{E}}_r(e^{-\int_0^T B_s^{\mu} ds })  \quad \text{for } t < T \label{limsupI}\\
	Y_T'&:= Y_T =\tilde{\mathcal{E}}_T(e^{-\int_0^T B_s^{\mu} ds })={\mathcal{E}}_T(e^{-\int_0^T B_s^{\mu} ds }). \label{limsupII}
\end{align}
\end{defi}
Under the assumption\footnote{Note, the assumption $\sup_{P \in {\mathcal{Z}}} E^{{P}}[e^{-\int_0^T B_s^{\mu} ds }] < \infty$ is always satisfied for $B^{\mu} > 0$ which is the case for a mortality intensity.}  $\sup_{{P} \in {\mathcal{Z}}} E^{{P}}[e^{-\int_0^T B_s^{\mu} ds }] < \infty$, the process $S^L$ is a c\`{a}dl\`{a}g $({P},\mathbb{F}^{*,{\mathcal{Z}}}_+)$- \linebreak supermartingale for all $P \in \mathcal{Z}$ by the proof of Theorem 3.2 in \cite{n_2015}. By the immersion property $S^L$ is also a $(\tilde{P},\mathbb{G}^{*,\tilde{\mathcal{Z}}}_+)$-supermartingale for all $\tilde{P} \in \tilde{\mathcal{Z}}$. Here, the filtration $\mathbb{G}^{{*,\tilde{\mathcal{Z}}}}:=(\mathcal{G}_t^{*,\tilde{\mathcal{Z}}})_{t \in [0,T]}$ (respectively $\mathbb{F}^{*,\mathcal{Z}}$) is defined similar as in $(\ref{filtrationP})$, i.e.,
\begin{equation}
	\mathcal{G}_t^{*,\tilde{\mathcal{Z}}}:= \mathcal{G}_t^* \vee \mathcal{N}_T^{\tilde{\mathcal{Z}}}, \quad t \in [0,T].
\end{equation}
It is quite standard to use the filtration $\mathbb{G}^{*,\tilde{\mathcal{Z}}}$ in the framework of model uncertainty, see e.g., \cite{n_2015}, \cite{nn_2013}. By considering the right-continuous version $\mathbb{G}^{*,\tilde{\mathcal{Z}}}_+$ of this filtration, we intuitively give the agent a ``little bit'' more information than the one available up to time $t$ at the market. For this reason $\mathbb{G}_+^{*,\tilde{\mathcal{Z}}}$ is not the natural filtration for financial applications. Hence, we discuss here further conditions to achieve more regularity on the paths of $(\tilde{\mathcal{E}}_t)_{t \in [0,T]}$. 
\begin{prop} \label{propModification}
Let Assumption \ref{assumptionnutzNew} hold for the families $(\mathcal{P}(t,\omega))_{(t,\omega) \in [0,T] \times \Omega_x}$ and $X$ be a nonnegative upper semianalytic function on $\Omega_x$. 
Furthermore, assume that for all $P \in \mathcal{P}, t \in [0,T]$ the set $\Phi_t^{P,X}:=\lbrace E^Q [X \vert \mathcal{F}_t]: Q \in \mathcal{P}(t;P) \rbrace$ is upward directed, i.e., for all $E^{Q_1} [X \vert \mathcal{F}_t], E^{Q_2} [X \vert \mathcal{F}_t] \in \Phi_t^{P,X}$ there exists $E^{Q}[X \vert \mathcal{F}_t] \in \Phi_t^{P,X}$ such that $E^{Q}[X \vert \mathcal{F}_t]=E^{Q_1} [X \vert \mathcal{F}_t] \vee E^{Q_2} [X \vert \mathcal{F}_t]$ $P$-a.s. Then the process $(\mathcal{E}_t(X))_{t \in [0,T]}$ has a c\`{a}dl\`{a}g $\mathcal{P}$-modification.
\end{prop}
As $\mathcal{E}(X)$ is a $(P,\mathbb{F}^*)$-supermartingale for all $P \in \mathcal{P}$, the existence of a c\`{a}dl\`{a}g $\mathcal{P}$-modification is equivalent to the fact that $t \mapsto E^P[\mathcal{E}_t(X)]$ is right-continuous for all $P \in \mathcal{P}$ \cite[p. 2037]{ns_2012}. Note, this is a generalization of Theorem 7 in \cite{protter} which states that a supermartingale $Y$ has a unique c\`{a}dl\`{a}g modification with respect to $P$ if and only if $(E^P[Y_t])_{t \in [0,T]}$ is right-continuous.

\begin{proof}
Fix $P \in \mathcal{P}$ and $t \in [0,T]$. Choose a sequence of $(t_n)_{n \in \mathbb{N}}$ with $t_n >t$ and $t_n \downarrow t$ for $n \to \infty$. We show that $t \mapsto E^P[\mathcal{E}_t(X)]$ is right-continuous for all $P \in \mathcal{P}$. The proof follows the idea of Proposition 4.3 in \cite{kramkov}.\\
1) In a first step we show that 
	\begin{equation}
		E^P \bigg[\esssup_{P' \in \mathcal{P}(t;P)} E^{P'}[X| \mathcal{F}_t] \bigg] = \sup_{P' \in \mathcal{P}(t;P)} E^P\big[E^{P'} [X \vert \mathcal{F}_t]\big]. \label{supremum-rep}
	\end{equation}
As we assumed that the set $\Phi_t^{P,X}$ is upward directed it follows by Theorem A.33 in \cite{follmer} that there exists an increasing sequence $(E^{Q_n} [X \vert \mathcal{F}_t])_n$ with $Q_n \in \mathcal{P}(t;P)$ in $\Phi_t^{P,X}$ such that $\lim_{n \to \infty} E^{Q_n} [X \vert \mathcal{F}_t]= \esssup_{P' \in \mathcal{P}(t;P)} E^{P'}[X| \mathcal{F}_{t}]$. Then by Fatou's Lemma we get
	\begin{align*}
	E^P \bigg[\esssup_{P' \in \mathcal{P}(t;P)} E^{P'}[X| \mathcal{F}_t] \bigg] &= E^P \bigg[\lim_{n \to \infty} E^{Q_n} [X \vert \mathcal{F}_t]\bigg] \leq \liminf_{n \to \infty} E^P \bigg[ E^{Q_n} [X \vert \mathcal{F}_t]\bigg ]\\ &
	\leq \sup_{P' \in \mathcal{P}(t;P)} E^P\big[E^{P'} [X \vert \mathcal{F}_t]\big].
	\end{align*} 
As it holds
	\begin{equation*}
	E^P \bigg[\esssup_{P' \in \mathcal{P}(t;P)} E^{P'}[X| \mathcal{F}_t] \bigg] \geq \sup_{P' \in \mathcal{P}(t;P)} E^P\big[E^{P'} [X \vert \mathcal{F}_t]\big],
	\end{equation*} 
we obtain the claim in $(\ref{supremum-rep})$.\\
2) As $P' \in \mathcal{P}(t;P)$ it holds for all $A \in \mathcal{F}_t$ that $E^P[\textbf{1}_A] = E^{P'}[\textbf{1}_A]$. Furthermore, every nonnegative $\mathcal{F}_t$-measurable random variable $Y$ can be approximated by simple functions by Sombrero's Lemma. So it follows $E^P[Y]=E^{P'}[Y]$ by monotone convergence.\\
Let $\epsilon > 0$. As $P_n \in \mathcal{P}(t_n;P) \subseteq  \mathcal{P}(t;P)$ for $t_n>t$, we have by $(\ref{supremum-rep})$ that
	\begin{align}
		E^P \bigg[\esssup_{P' \in \mathcal{P}(t;P)} E^{P'}[X| \mathcal{F}_t] \bigg] & < E^P\big[E^{P_n} [X \vert \mathcal{F}_t]\big] + \epsilon		= E^{P_n} \big[E^{P_n} [X \vert \mathcal{F}_t]\big] + \epsilon= E^{P_n}[X]+ \epsilon \nonumber\\
		&= E^{P_n} \big[E^{P_n} [X \vert \mathcal{F}_{t_n}]\big] + \epsilon  
		= E^{P} \big[E^{P_n} [X \vert \mathcal{F}_{t_n}]\big] + \epsilon \nonumber \\
		&\leq \sup_{P' \in \mathcal{P}(t_n;P)} E^{P}\big[ E^{P'} [X \vert \mathcal{F}_{t_n}]\big] + \epsilon  \nonumber \\
		&= E^P \bigg[\esssup_{P' \in \mathcal{P}(t_n;P)} E^{P'}[X| \mathcal{F}_{t_n}] \bigg]  + \epsilon. \nonumber
	\end{align}
3) Due to the supermartingale property of $\mathcal{E}(X)$ the expectation is decreasing such that 
	\begin{equation*}
		E^P \bigg[\esssup_{P' \in \mathcal{P}(t;P)} E^{P'}[X| \mathcal{F}_t] \bigg] \geq \lim_{n \to \infty}	E^P \bigg[\esssup_{P' \in \mathcal{P}(t_n;P)} E^{P'}[X| \mathcal{F}_{t_n}] \bigg],
	\end{equation*}
which finishes the proof.
\end{proof}
In the proof of Lemma 3.4 in \cite{ns_2012} it is used that the set $\Phi_t^{P,X}$ is upward directed for all $t \in [0,T], P \in \mathcal{P}$ and $X$ $\mathcal{F}_T$-measurable with $\sup_{P \in \mathcal{P}} E^P[\vert X \vert] < \infty$ if $\mathcal{P}$ is stable under pasting, i.e., for all $P \in \mathcal{P}, \tau \ \mathbb{F}$-stopping time, $\Lambda \in \mathcal{F}_{\tau}, P_1,P_2 \in \mathcal{P}(\mathcal{F}_{\tau};P)$ the measure 
	\begin{equation}
		\overline{P}(A)=:E^P[P_1(A \vert \mathcal{F}_{\tau})\textbf{1}_{\Lambda} + P_2(A \vert \mathcal{F}_{\tau})\textbf{1}_{\Lambda^c}], \quad A \in \mathcal{F} \label{forkStability}
	\end{equation}
is again an element of $\mathcal{P}$. Note that condition $(\ref{forkStability})$ is not the same concept as the stability under pasting of Assumption \ref{assumptionnutzNew}. 	

\begin{remark} 
In Proposition \ref{propModification} we first choose the contingent claim $X$ in which we are interested in and then assume that only for this fixed $X$ the set $\Phi_t^{P,X}$ is upward directed. However, if $\mathcal{P}$ satisfies the property of stability under pasting in the sense of $(\ref{forkStability})$, it follows that $\Phi_t^{X,P}$ is upward directed for all $P \in \mathcal{P}$ and any $\mathcal{F}_T$-measurable random variables $X$ with $\sup_{P \in \mathcal{P}} E^P[\vert X \vert] < \infty$. Thus for $X$ nonnegative, upper semianaliytic and $\mathcal{F}_T$-measurable the assumptions in Proposition \ref{propModification} are always satisfied if $\mathcal{P}$ is stable under pasting in the sense of $(\ref{forkStability})$.
\end{remark}
We now provide an example of families of priors for which $(\ref{forkStability})$ is satisfied. 
\begin{prop} \label{G_Expectation_Fork}
	Consider the setting of Example $\ref{Example2RandomG-expectation}$ with the families $(\mathcal{P}_{\textbf{D}}(t,\omega))_{(t,\omega) \in [0,T] \times \Omega_x}$ defined as in $(\ref{Set_random_G})$. Furthermore, let Assumption \ref{AssumptionG-Expectation} hold and the process $\textbf{D}=(\textbf{D}_t)_{t \in [0,T]}$ be given for fixed $\underline{a}^0, \overline{a}^0, \underline{a}^1, \overline{a}^1$. Then the set $\mathcal{P}_{\textbf{D}}$ satisfies $(\ref{forkStability})$ and thus the set $\Phi_t^{P,X}$ is upward directed for all $t \in [0,T], P \in \mathcal{P}_{\textbf{D}}$ and $X$ $\mathcal{F}_T$-measurable. 
\end{prop}
\begin{proof}
	We prove this result in several steps. From now on let $P \in \mathcal{P}_{\textbf{D}}, \tau \in [0,T]$ a $\mathbb{F}$-stopping time, $\Lambda \in \mathcal{F}_{\tau}$ and $P_1,P_2 \in \mathcal{P}(\tau;P)$. In addition, consider $\overline{P}$ defined as in $(\ref{forkStability})$.  \\
1) We first show that 
	\begin{equation}
			\overline{P}=P \text{ on } \mathcal{F}_{\tau}. \label{Step1}
	\end{equation}
Let $A \in \mathcal{F}_{\tau}$, then
	\begin{align*}
		\overline{P}(A)= E^P[P_1(A \vert \mathcal{F}_{\tau})\textbf{1}_{\Lambda} + P_2(A \vert \mathcal{F}_{\tau})\textbf{1}_{\Lambda^c}] =E^P [\textbf{1}_A \textbf{1}_{\Lambda} + \textbf{1}_A \textbf{1}_{\Lambda^c}] = P(A).
	\end{align*}
Moreover, for $A \in \mathcal{F}$ we have 
	\begin{align}
		\overline{P}(A) &= E^P[P_1(A \vert \mathcal{F}_{\tau})\textbf{1}_{\Lambda} + P_2(A \vert \mathcal{F}_{\tau})\textbf{1}_{\Lambda^c}] \
		= E^{P_1}[E^{P_1}[\textbf{1}_{A \cap \Lambda} \vert \mathcal{F}_{\tau}]] + E^{P_2}[E^{P_2}[\textbf{1}_{A \cap \Lambda^c} \vert \mathcal{F}_{\tau}]] \nonumber  \\
		&= E^{P_1}[\textbf{1}_{A \cap \Lambda}] + E^{P_2}[\textbf{1}_{A \cap \Lambda^c}]= P_1(A \cap \Lambda) + P_2 (A \cap \Lambda^c). \label{Prewritten}
	\end{align}	
2) By using $(\ref{Prewritten})$ for a random variable $X$ on $\Omega_x$ and $t \in [0,T]$ we have
\begin{equation}
	E^{\overline{P}}[X \textbf{1}_{\Lambda} \vert  \mathcal{F}_t]=E^{P_1}[X \textbf{1}_{\Lambda} \vert \mathcal{F}_t] \quad \text{and} \quad 
	E^{\overline{P}}[X \textbf{1}_{\Lambda^c} \vert \mathcal{F}_t]=E^{P_2}[X \textbf{1}_{\Lambda^c} \vert \mathcal{F}_t] \label{Part4II}.
\end{equation}

Next we show that $B^{S}, B^{\mu}$ are local $(\overline{P},\mathbb{F})$-martingales. We prove it for a (local) $(Q,\mathbb{F})$-martingale $X=(X_t)_{t \in [0,T]}$ for $Q \in \lbrace P,P_1, P_2 \rbrace$. We have to distinguish two cases.
Combining the two equations in $(\ref{Part4II})$ we get for $0 \leq \tau \leq t \leq T$ and $t \leq s$
\begin{align}
	E^{\overline{P}}[X_s \vert \mathcal{F}_t]=& E^{\overline{P}}[X_s \textbf{1}_{\Lambda} \vert \mathcal{F}_t] + E^{\overline{P}}[X_s \textbf{1}_{\Lambda^c} \vert \mathcal{F}_t] 
	=E^{{P_1}}[X_s \textbf{1}_{\Lambda} \vert \mathcal{F}_t] + E^{{P_2}}[X_s \textbf{1}_{\Lambda^c} \vert \mathcal{F}_t] \nonumber \\
	=& E^{{P_1}}[X_s  \vert \mathcal{F}_t]\textbf{1}_{\Lambda} + E^{{P_2}}[X_s \vert \mathcal{F}_t]\textbf{1}_{\Lambda^c}=X_t. \label{MartingaleI}
\end{align}
For the case $0 \leq t \leq \tau \leq T$ and $t \leq s$ it holds
\begin{align}
	E^{\overline{P}}[X_s \textbf{1}_{\Lambda} \vert  \mathcal{F}_t]&= E^{P_1}[X_s\textbf{1}_{\Lambda} \vert  \mathcal{F}_t]= E^{P_1}[E^{P_1}[X_s \textbf{1}_{\Lambda} \vert \mathcal{F}_{\tau}] \vert \mathcal{F}_{t}]= E^{P_1}[\textbf{1}_{\Lambda}E^{P_1}[X_s \vert \mathcal{F}_{\tau}] \vert \mathcal{F}_{t}] \nonumber \\
	&=E^{P_1}[\textbf{1}_{\Lambda} X_{\tau} \vert \mathcal{F}_t]=E^{P}[X_{\tau} \textbf{1}_{\Lambda} \vert \mathcal{F}_t]. \label{Splitting}
\end{align}
Here, we used in the last step that for $t \leq \tau$ and $X \in \mathcal{F}_{\tau}$ we have $E^{P}[X \textbf{1}_{\Lambda} \vert \mathcal{F}_t]= E^{P_1}[X \textbf{1}_{\Lambda} \vert \mathcal{F}_t]$. In the same way as in $(\ref{Splitting})$ we can derive
\begin{equation*}
	E^{P_2}[X_s \textbf{1}_{\Lambda^c} \vert \mathcal{F}_{t}]= E^{P}[X_{\tau} \textbf{1}_{\Lambda^c} \vert \mathcal{F}_t], 
\end{equation*}
which implies with $(\ref{Part4II})$ and $(\ref{Splitting})$ that
\begin{equation*}
	E^{\overline{P}}[X_s \vert \mathcal{F}_t]= E^{{P}_1}[X_s \textbf{1}_{\Lambda} \vert  \mathcal{F}_t]+E^{P_2}[X_s \textbf{1}_{\Lambda^c} \vert \mathcal{F}_{t}]= E^{{P}}[X_s \textbf{1}_{\Lambda} \vert  \mathcal{F}_t]+E^{P}[X_s \textbf{1}_{\Lambda^c} \vert \mathcal{F}_{t}]= X_t.
\end{equation*}
Thus, $B^{\mu}$ and $B^{S}$ are local $(\overline{P},\mathbb{F})$-martingales. \\
3) We show
	\begin{equation}
		\alpha^{\overline{P}}_t(\omega):= d \langle B \rangle_t^{\overline{P}} / dt(\omega) = \textbf{1}_{[0, \tau]}(t) \alpha^P_t(\omega) + \textbf{1}_{]\tau,T]}(t)( \alpha^{P_1}_t(\omega) \textbf{1}_{\Lambda}(\omega) + \alpha^{P_2}_t(\omega)\textbf{1}_{\Lambda^c}(\omega)) \label{FormulaQuadraticVariation}
	\end{equation}
with $\alpha^Q_t:=d \langle B \rangle_t^{Q} / dt$ for $Q \in \lbrace P, P_1,P_2 \rbrace $. 
First, we prove
	\begin{equation}
		\langle B \rangle^{\overline{P}}_t(\omega)= \langle B \rangle^{P_1}_{t}(\omega) \textbf{1}_{\Lambda}(\omega) + \langle B \rangle^{P_2}_{t}(\omega) \textbf{1}_{\Lambda^c}(\omega). \label{QuadraticVariation}
	\end{equation}
Consider a partition $\pi: 0=t_0 < t_1 <...<  t_n=t$ of the interval $[0,t]$ with mesh size $\| \pi \|:= \max \lbrace \vert t_k - t_{k-1} \vert: k = 1,...,n \rbrace$. 
Set $\Delta_{t_k}^2:=(B_{t_k}-B_{t_{k-1}})^2$ for $k \in \mathbb{N}$. Then it holds
	\begin{align*}
		&\overline{P} \big( \lbrace \omega \in \Omega: \vert \sum_{k=0}^n \Delta_{t_k}^2(\omega)-\langle B \rangle^{P_1}_{t}(\omega) \textbf{1}_{\Lambda}(\omega) - \langle B \rangle^{P_2}_{{t}}(\omega) \textbf{1}_{\Lambda^c}(\omega) \vert > \epsilon \rbrace \big) = \\
		&\overline{P} \big( \lbrace \omega \in \Lambda: \vert \sum_{k=0}^n \Delta_{t_k}^2(\omega)-\langle B \rangle^{P_1}_{{t}}(\omega) \vert  > \epsilon \rbrace \big) + 
				\overline{P} \big( \lbrace \omega \in \Lambda^c: \vert \sum_{k=0}^n\Delta_{t_k}^2(\omega)-\langle B \rangle^{P_2}_{{t}}(\omega)   \vert > \epsilon \rbrace \big) = \\
		&P_1 \big( \lbrace \omega \in \Lambda: \vert \sum_{k=0}^n \Delta_{t_k}^2(\omega)-\langle B \rangle^{P_1}_{{t}}(\omega)\vert > \epsilon \rbrace \big) + 
				P_2 \big( \lbrace \omega \in \Lambda^c: \vert \sum_{k=0}^n \Delta_{t_k}^2(\omega)-\langle B \rangle^{P_2}_{{t}}(\omega)\vert > \epsilon \rbrace \big) \leq \\
		&P_1 \big( \lbrace \omega \in \Omega: \vert \sum_{k=0}^n \Delta_{t_k}^2(\omega)-\langle B \rangle^{P_1}_{{t}}(\omega)\vert > \epsilon \rbrace \big) + 
				P_2 \big( \lbrace \omega \in \Omega: \vert \sum_{k=0}^n\Delta_{t_k}^2-\langle B \rangle^{P_2}_{{t}}(\omega)\vert > \epsilon \rbrace \big) \longrightarrow_{\| \pi \| \to 0} 0,
	\end{align*}
where we used $(\ref{Prewritten})$. Thus $(\ref{QuadraticVariation})$ follows as the limit of convergence in probability is almost surely unique. 
If the quadratic variation for the process $X$ with respect to $P$ exists, then $X$ has the same quadratic variation with respect to all probability measures $Q \sim P$ \cite[p. 15]{rheinlaender}. As $P=\overline{P}$ on $\mathcal{F}_{\tau}$, it follows from step 1) that $\langle B \rangle^{\overline{P}}_t= \langle B \rangle^{{P}}_t$ for $t \in [0,\tau]$. By putting all these facts together we can conclude that $(\ref{FormulaQuadraticVariation})$ holds. Furthermore, as $\alpha^{P}, \alpha^{P_i}, i=1,2$ take values in $\textbf{D}$, it follows that also $\alpha^{\overline{P}}$ take values in this set, i.e., $\overline{P} \in \mathcal{P}_{\textbf{D}}$.   
\end{proof}
%*************************Thesis**************************(Okay from Francesca)
%The next result is a straight forward consequence of Lemma \ref{G_Expectation_Fork}.  
%\begin{cor}
	%The set $\mathcal{P}_{\textbf{D}}$ defined in Example \ref{Example1G-expectation} satisfies $(\ref{forkStability})$.
%\end{cor}
%*********************************************************\\
Next, we show that the property of $\Phi_t^{P,X}$ being upward directed, which is a property on $(\Omega_x, \mathcal{F})$, can be transferred to the extended space $(\tilde{\Omega}_x, \mathcal{G})$.
		
\begin{prop} \label{heritage}
	Let Assumption \ref{assumptionnutzNew} hold for the families $(\mathcal{P}(t,\omega))_{(t,\omega) \in [0,T] \times \Omega_x}$. Assume that for every nonnegative upper semianalytic function $X$ on $\Omega_x$, $t \in [0,T]$ and $P \in \mathcal{P}$ the set $\Phi_t^{P,X}:=\lbrace E^Q[X \vert \mathcal{F}_t]: Q \in {\mathcal{P}}(t;P) \rbrace $ is upward directed. Then for every nonnegative upper semianalytic function $\tilde{X}$ on $\tilde{\Omega}_x$ which is $\mathcal{G}_T^{\mathcal{P}}$-measurable, $t \in [0,T]$ and $\tilde{P} \in \tilde{\mathcal{P}}$ the set $\tilde{\Phi}_t^{\tilde{P},X}:=\lbrace E^{\tilde{Q}} [\tilde{X} \vert \mathcal{G}_t]: \tilde{Q} \in \tilde{\mathcal{P}}(t;\tilde{P}) \rbrace \rbrace$ is upward directed with $\tilde{P}:=P \otimes \hat{P}$.
\end{prop}
		
\begin{proof}
Let $P \in \mathcal{P}$, $t \in [0,T]$ and $X$ nonnegative upper semianalytic such that the corresponding set $\Phi_t^{P,X}$ is upward directed. Consider $\tilde{P}= {P} \otimes \hat{P} \in \tilde{\mathcal{P}}$ and $\tilde{Q}_1 \in \tilde{\mathcal{P}}(t;\tilde{P})$. As $\tilde{Q}_1 \in \tilde{\mathcal{P}}$ there exists $Q_1 \in \mathcal{P}$ such that $\tilde{Q}_1 = Q_1 \otimes \hat{P}$. Next, we prove the following statement \\
	\begin{equation}
		\tilde{Q}_1 \in \tilde{\mathcal{P}}(t;\tilde{P}) \text{ if and only if } Q_1 \in \mathcal{P}(t;P), \label{claim}
	\end{equation}
	is equivalent to show that
	\begin{equation*}
	\tilde{Q}_1(\tilde{A})=\tilde{P}(\tilde{A}) \quad \forall \tilde{A} \in \mathcal{G}_t \Longleftrightarrow {Q}_1({A})={P}({A}) \quad \forall {A} \in \mathcal{F}_t.
	\end{equation*}
	Let $A \in \mathcal{F}_t \subseteq \mathcal{G}_t$ then we have $Q_1(A)=Q_1 \otimes \hat{P}(A)= \tilde{Q}_1(A)= \tilde{P}(A) =P \otimes \hat{P}(A) =P(A)$, which shows the first implication. For the other direction take $\tilde{A} \in \mathcal{G}_t$ and use Lemma 2.12 in \cite{bz_2019} such that we have
	\begin{equation*}
		\tilde{Q}_1(\tilde{A})= E^{\tilde{Q}_1}[\textbf{1}_{\tilde{A}}]=E^{Q_1}[E^{\hat{P}} [\textbf{1}_{\tilde{A}}]] = E^{P}[E^{\hat{P}} [\textbf{1}_{\tilde{A}}]] = E^{\tilde{P}}[\textbf{1}_{\tilde{A}}]= \tilde{P}(\tilde{A}),
	\end{equation*}
	which proves the claim in $(\ref{claim})$. 
	Now consider also $\tilde{Q}_2 \in \mathcal{P}(t;\tilde{P})$, i.e., $\tilde{Q}_2 = Q_2 \otimes \hat{P}$. 
	As $\Phi_t^{P,X}$ is upward directed we know that there exists $Q \in \mathcal{P}(t;P)$ such that
	\begin{equation}
		E^{Q}[X \vert \mathcal{F}_t]=E^{Q_1} [X \vert \mathcal{F}_t] \vee E^{Q_2} [X \vert \mathcal{F}_t] \quad P\text{-a.s.}\label{PhiUpper}
	\end{equation} for any $X$ nonnegative and upper semianalytic.
 	Set $\tilde{Q}:=Q \otimes \hat{P}$. We now show that $\tilde{Q} \in \tilde{\mathcal{P}}(t; \tilde{P})$ and $E^{\tilde{Q}} [\tilde{X} \vert \mathcal{G}_t]=E^{\tilde{Q}_1} [\tilde{X} \vert \mathcal{G}_t] \vee E^{\tilde{Q}_2} [\tilde{X} \vert \mathcal{G}_t]$ with $\tilde{Q}_i:= Q_i \otimes \hat{P}$ for $i=1,2$. The first property follows by $(\ref{claim})$. By Proposition 2.16 in \cite{bz_2019} we know that for $t \geq 0, \tilde{Q}=Q \otimes \hat{P}$ and $\tilde{X}$ nonnegative and $\mathcal{G}^{\mathcal{P}}_T$-measurable
	\begin{align*}
		E^{\tilde{Q}}[\tilde{X} \vert \mathcal{G}_t] &= \textbf{1}_{\lbrace \tilde{\tau} \leq t \rbrace} E^Q[\varphi(x, \cdot) \vert \mathcal{F}_t] \big \vert_{x = \tilde{\tau}} + \textbf{1}_{\lbrace \tilde{\tau} > t \rbrace} e^{\Gamma_t} E^Q [E^{\hat{P}} [\textbf{1}_{\lbrace \tilde{\tau} > t \rbrace} \tilde{X}] \vert \mathcal{F}_t]   \\
		&=	\textbf{1}_{\lbrace \tilde{\tau} \leq t \rbrace} \big( E^{Q_1}[\varphi(x, \cdot) \vert \mathcal{F}_t] \vee E^{Q_2}[\varphi(x, \cdot) \vert \mathcal{F}_t] \big \vert_{x = \tilde{\tau}} \big) + 
				\\ 
		& +\textbf{1}_{\lbrace \tilde{\tau} > t \rbrace} e^{\Gamma_t} \big(E^{Q_1} [E^{\hat{P}} [\textbf{1}_{\lbrace \tilde{\tau} > t \rbrace} \tilde{X}] \vert \mathcal{F}_t] \vee E^{Q_2} [E^{\hat{P}} [\textbf{1}_{\lbrace \tilde{\tau} > t \rbrace} \tilde{X}] \vert \mathcal{F}_t]\big) \\		
		&= E^{\tilde{Q}_1}[\tilde{X} \vert \mathcal{G}_t]  \vee E^{\tilde{Q}_2}[\tilde{X} \vert \mathcal{G}_t] \quad \tilde{Q} \text{-a.s.}
	\end{align*}
	with $\varphi$ as in $(\ref{varphi})$. Here, we used $(\ref{PhiUpper})$ which is possible as $E^{\hat{P}} [\textbf{1}_{\lbrace \tilde{\tau} > t \rbrace} \tilde{X}]$ is nonnegative.    
\end{proof}
		
\begin{lemma}\label{lemmaModification}
	Let Assumption \ref{assumptionnutzNew} hold for the families $(\mathcal{P}(t,\omega))_{(t,\omega) \in [0,T] \times \Omega_x}$ and $\tilde{X}$ be a nonnegative upper semianalytic function on $\tilde{\Omega}_x$ which is $\mathcal{G}_T^{\mathcal{P}}$-measurable.	 Furthermore, assume that for every nonnegative upper semianalytic function $X$ on $\Omega_x$, $t \in [0,T]$ and $P \in \mathcal{P}$ the set $\Phi_t^{P,X}:=\lbrace E^Q[X \vert \mathcal{F}_t]: Q \in {\mathcal{P}}(t;P) \rbrace $ is upward directed. Then the process $(\tilde{\mathcal{E}}_t(X))_{t \in [0,T]}$ has a c\`{a}dl\`{a}g $\tilde{\mathcal{P}}$-modification.
\end{lemma}
		
\begin{proof}
	Proposition $\ref{heritage}$ allows us to transfer the property of upward directed from the set $\Phi_t^{P,X}$ to $\tilde{\Phi}_t^{\tilde{P}, \tilde{X}}$. Then we can use exactly the same arguments in Proposition \ref{propModification} which is possible as the sublinear conditional operator $\tilde{\mathcal{E}}$ admits the representation as essential supremum in $(\ref{esssupBiagini})$ under Assumption \ref{assumptionnutzNew}.
\end{proof}

The results in Proposition \ref{propModification}, \ref{heritage} and \ref{lemmaModification} are also valid by replacing $\Omega_x$ defined as in Subsection \ref{SubsectionMarketModel} by $\Omega=C_0(\mathbb{R}_+,\mathbb{R}^d)$ or $\Omega=D_0(\mathbb{R}_+,\mathbb{R}^{d})$ for $d \in \mathbb{N}$. \\
One advantage by working with the c\`{a}dl\`{a}g $\mathcal{P}$-modification compared with the Definition \ref{defiLongevityBond} is that we get the path regularity without being forced to consider a process adapted to the right-continuous version of a filtration. Nevertheless, if $\mathcal{P}$ does not allow the existence of c\`{a}dl\`{a}g $\mathcal{P}$-modification, the approach in Definition \ref{defiLongevityBond} always guarantees c\`{a}dl\`{a}g paths.

\begin{subsection}{Numerical valuation}
In the sequel we derive a numerical representation of the longevity bond $S^L$ introduced in Definition \ref{defiLongevityBond} by using the affine structure with parameter uncertainty of the underlying intensity. This is possible by generalizing Theorem 6.2 in \cite{fns_2019}.\\
We define the upper bounds for the intervals $a^*(x)$ and $b^*(x)$ in $(\ref{affineBounds})$ which are given by
\begin{align}
	\overline{a}(x) =\overline{a}^0 + \overline{a}^1 x^+ \quad \text{ and } \quad \overline{b}(x) =\overline{b}^0 + \underbrace{\underline{b}^1\textbf{1}_{\lbrace x <0 \rbrace} + \overline{b}^1 \textbf{1}_{\lbrace x \geq 0 \rbrace}}_{:=\overline{B}^{1,x}}.
\end{align}
\begin{prop} \label{GeneralizedNumericalRepr}
Assume that for all $P \in \mathcal{Z}$
	\begin{equation*}
		\beta_t^P \leq \overline{b}^0 + \overline{B}^{1,x} B_t^{\mu},
	\end{equation*}
	$dP \otimes dt$-almost everywhere for $0 \leq t \leq T$. Moreover, assume either that $\underline{a}^1=\overline{a}^1=0$ or that for all $P \in \mathcal{Z}$, $B_t^{\mu} \geq 0$ $P \otimes dt$-a.e. Furthermore, there exists $\overline{P} \in \mathcal{Z}$ and a one-dimensional $(\overline{P}$,$\mathbb{F})$-Brownian motion $W$ such that the componentwise canonical process $B^{\mu}$ under $\overline{P}$ is the unique strong solution of 
	\begin{equation} \label{affineMarkovian}
		dB_t^{\mu} = (\overline{b}^0 + \overline{B}^{1,x} B_t^{\mu})dt + \sqrt{ \overline{a}(B_t^{\mu})}dW_t, \quad B_0^{\mu}=\omega_0^{\mu}.
	\end{equation}
	Then, for all $u \geq 0$ and $0 \leq t \leq T$ 
	\begin{equation}
		\mathcal{E}_t(e^{-\int_t^T B_s^{\mu}ds})=\esssupO_{P' \in \mathcal{Z}(t;\overline{P})} E^{P'}\big[e^{-\int_t^T B_s^{\mu} ds}\big\vert \mathcal{F}_t]= \exp(\phi(T-t,0) + \psi(T-t,0)B_t^{\mu}) \quad \overline{P} \text{-a.s.}, \label{OperatorcondSolution}
	\end{equation}
	where $\phi$ and $\psi$ solve the Riccati equations
	\begin{align*}
	\partial_t \phi(t,u)=\frac{1}{2} \overline{a}^0 \psi(t,u)^2 + \overline{b}^0 \psi(t,u) \quad \phi(0,u)=0 \\
	\partial_t \psi(t,u)=\frac{1}{2} \overline{a}^1 \phi(t,u)^2 + \overline{B}^{1,x} \phi(t,u)-1 \quad \psi(0,u)=u.
	\end{align*}
\end{prop}
\begin{proof}
1) Let $P \in \mathcal{Z}$. With the same arguments as in Proposition 6.2 in \cite{fns_2019} we have 
	\begin{equation*}
		E^P[e^{-\int_0^t B_s^{\mu}ds}] \leq E^{\overline{P}}[e^{-\int_0^t B_s^{\mu}ds}],
	\end{equation*}
	where $\overline{P}$ is given by the assumptions of the Proposition. 
As $P \in \mathcal{Z}$ is arbitrary and $\overline{P} \in \mathcal{Z}$ it follows
	\begin{equation} \label{comparison}
		E^{\overline{P}}[e^{-\int_0^T B_s^{\mu}ds}] = \sup_{P \in \mathcal{Z}}E^{P}[e^{-\int_0^T B_s^{\mu}ds}]. 
		\end{equation}
2) We now show that
	\begin{equation} \label{comparisonIntermediate}
		E^{\overline{P}}[e^{-\int_0^T B_s^{\mu}ds}\vert \mathcal{F}_t] = \esssupO_{P' \in \mathcal{Z}(t;\overline{P})} E^{P'}\big[e^{-\int_0^T B_s^{\mu} ds}\big\vert \mathcal{F}_t] \quad \overline{P}\text{-a.s.} 
	\end{equation}
As $\overline{P} \in \mathcal{Z}(t;\overline{P}) \subseteq \mathcal{Z}$, the inequality $ E^{\overline{P}}[e^{-\int_0^T B_s^{\mu}ds}\vert \mathcal{F}_t] \leq \esssupO_{P' \in \mathcal{Z}(t,\overline{P})} E^{P'}\big[e^{-\int_t^T B_s^{\mu} ds}\big\vert \mathcal{F}_t]$ follows directly.
For the other direction we show that for all $P' \in \mathcal{Z}(t;\overline{P})$ 
	\begin{equation*}
	E^{\overline{P}}[E^{\overline{P}}[e^{-\int_0^T B_s^{\mu}ds}\vert \mathcal{F}_t]] \geq E^{\overline{P}}[E^{P'}[e^{-\int_0^T B_s^{\mu}ds}\vert \mathcal{F}_t]].
	\end{equation*}
Fix $P' \in \mathcal{Z}(t;\overline{P})$, then by $(\ref{comparison})$ we have
\begin{align*}
E^{\overline{P}}[E^{\overline{P}}[e^{-\int_0^T B_s^{\mu}ds}\vert \mathcal{F}_t]]&=E^{\overline{P}}[e^{-\int_0^T B_s^{\mu}ds}]= \sup_{P \in \mathcal{Z}} E^P[e^{-\int_0^T B_s^{\mu}ds}] \geq  \sup_{P' \in \mathcal{Z}(t;\overline{P})} E^{P'}[e^{-\int_0^T B_s^{\mu}ds}] \\
&\geq E^{P'}[e^{-\int_0^T B_s^{\mu}ds}]= E^{P'}[E^{P'}[e^{-\int_0^T B_s^{\mu}ds}\vert \mathcal{F}_t]]=E^{\overline{P}}[E^{P'}[e^{-\int_0^T B_s^{\mu}ds}\vert \mathcal{F}_t]],
\end{align*}
where we use in the last equality that $P'=\overline{P}$ on $\mathcal{F}_t$. \\
3) As $B^{\mu}$ is an affine process in the classical sense, we get by Theorem 10.14 in \cite{f_2009} and $(\ref{comparisonIntermediate})$ the representation via Riccati equations as in $(\ref{OperatorcondSolution})$. 
\end{proof}

By Definition \ref{defiLongevityBond} the value process of the longevity bond $S^L_t=\limsup_{r \downarrow t, r \in \mathbb{Q}} Y_r \textbf{1}_{N^c}$, $t \in [0,T)$ can be rewritten by $(\ref{OperatorcondSolution})$ and the corresponding Riccati equations as
\begin{align}
	Y_r=& \tilde{\mathcal{E}}_r(e^{-\int_0^T B_s^{\mu}ds})=\mathcal{E}_r(e^{-\int_0^T B_s^{\mu}ds})= \esssup_{P' \in \mathcal{Z}(r;P)} E^{P'}\big[e^{-\int_0^T B_s^{\mu} ds}\big\vert \mathcal{F}_r] \nonumber \\
	=&(e^{-\int_0^r B_s^{\mu} ds} )\esssup_{P' \in \mathcal{Z}(r;P)} E^{P'}\big[e^{-\int_r^T B_s^{\mu} ds}\big\vert \mathcal{F}_r] \nonumber \\
	=&(e^{-\int_0^r B_s^{\mu} ds} )\exp(\phi(T-r,0) + \psi(T-r,0)B_r^{\mu}). \label{LongevityRiccati}
\end{align}

As already mentioned in Remark 6.3 in \cite{fns_2019} there are two important cases in which the assumptions of Proposition \ref{GeneralizedNumericalRepr} are satisfied.
\begin{enumerate}
\itemsep0pt
	\item Non-linear Vasicek model with state space $\mathcal{O}=\mathbb{R}$, i.e., $\Theta=[\underline{b}^0,\overline{b}^0] \times \lbrace \underline{b}^1 \rbrace \times [\underline{a}^0,\overline{a}^0] \times \lbrace 0 \rbrace$ with $\underline{b}^1=\overline{b}^1$. 
	\item Non-linear CIR model with state space $\mathcal{O}=\mathbb{R}_{>0}$, i.e., $\Theta=[\underline{b}^0,\overline{b}^0] \times [\underline{b}^1,\overline{b}^1] \times \lbrace 0 \rbrace \times [\underline{a}^1,\overline{a}^1]$ with $\underline{b}^0 \geq \overline{a}^1/2$.
\end{enumerate}

\begin{remark}
	In \cite{luciano_vigna_2005} non-mean reverting processes are suggested as they better fit observed data on mortality intensity. As the Vasicek and the CIR model also belong to classes of Ornstein-Uhlenbeck and Feller processes, we can include this non-mean reverting property to the correspondent non-linear cases by setting $\underline{b}^0=\overline{b}^0=0$. 
\end{remark}

We now consider the valuation of a contingent claim $f(S_T)$ in a setting for given families of probability measures $(\tilde{\mathcal{A}}(t,y))_{(t,y) \in [0,T] \times \mathbb{R}^2} \subseteq \mathcal{P}(\Omega_x)$ with $y=(y^{\mu},y^S)$ and a Lipschitz function $f: \mathcal{O}^S \to \mathbb{R}_+$, $\mathcal{O}^S \subseteq \mathbb{R}$.
We want to find a way to numerically compute the following value function $v: [0,T] \times \mathcal{O} \to \mathbb{R}$, $\mathcal{O} \subseteq \mathbb{R}^2$
\begin{align}
	v(t,y)&:=\sup_{P \in \tilde{\mathcal{A}}(t,y)} E^P[e^{-\int_t^T B^{\mu}_s ds} f(S_T) \vert B_t^{\mu}=y^{\mu}, S_t=y^S]. \label{defiValueFuct}
\end{align}
We construct an example for a space $\Omega_x$ and families of probability measures $(\tilde{\mathcal{A}}(t,y))_{(t,y) \in [0,T] \times \mathbb{R}^2}$ such that the value function $v(t,y)$ in $(\ref{defiValueFuct})$ can be explicitly computed via generalized Riccati equations and PDEs. More generalized cases are studied in \cite{bahar_francesca_katharina}. 
\begin{example}
Set $(\Omega_x, \mathcal{F}):=(\Omega^{\mu} \times \Omega^{S}, \mathcal{F}^{\mu} \otimes \mathcal{F}^S)$ with $x:=(x^{\mu}, x^S) \in \mathbb{R}^2$, $\Omega^{\mu}:=C_{x^{\mu}}([0,T], \mathbb{R})$ and $\Omega^{S}:=C_{x^{S}}([0,T], \mathbb{R})$ equipped with the Borel $\sigma$-algebra $\mathcal{F}^{\mu}:=\mathcal{B}(\Omega^{\mu})$ and $\mathcal{F}^{S}:=\mathcal{B}(\Omega^{S})$ respectively. The canonical processes on $\Omega^{\mu}$ and $\Omega^S$ are denoted by $B^{\mu}$ and $B^S$, respectively. For $t \in [0,T]$, $y=(y^{\mu},y^S) \in \mathcal{O}^{\mu} \times \mathbb{R}$, we consider on $\Omega_x$ the following family of probability measures 
	\begin{equation}
		\tilde{\mathcal{A}}(t,y):=\lbrace P=P^{\mu} \otimes P^{S}: P^{\mu} \in \mathcal{A}(t,y^{\mu},\Theta^{\mu}), P^{S} \in \mathcal{P}^S \rbrace \subseteq \mathcal{P}(\Omega_x),  \label{ProductSet}
	\end{equation}
	where $\mathcal{P}^S$ is the weakly compact set of probability measures representing the $G$-expectation as an upper expectation on $\Omega^S$ as in Theorem 2.5, Chapter VI in \cite{shige_script}. Let $(\mathcal{A}(t,y^{\mu},\Theta^{\mu}))_{(t,y^{\mu}) \in [0,T] \times \mathcal{O}^{\mu}}$ be proper families of probability measures on $\Omega^{\mu}$ with state space $\mathcal{O}^{\mu} \subseteq \mathbb{R}$ as in Definition \ref{defiAffineDominated}. Moreover, assume that the asset price $S=(S_s)_{s \in [t,T]}$ on $\Omega^S$ satisfies the following SDE
\begin{align*}
	dS_s&=b(S_s)d_s + h(S_s)d\langle {B^S} \rangle_s + \sigma(S_s)d B^S, \quad s \in [t,T]\\
	S_t&=y^{S},
\end{align*}
where the canonical process $B^S$ is a one-dimensional $G$-Brownian motion on $\Omega_S$ due to the definition of $\mathcal{P}^S$ and $b,h,\sigma: \mathbb{R} \to \mathbb{R}$ are Lipschitz continuous functions. Then $v^S(t,y^S)$ is the unique viscosity solution of the following PDE
\begin{align*}
	&\partial_t v^{S} + F(D^2 v^S, D v^S, v^S, y^S)=0 \\
	&v^S(T,y^S)=f(y^S)
\end{align*}
with 
\begin{equation*}
	F(D^2 v^S, D v^S, v^S, y^S)= G(\sigma(y^S)^2 D^2 v^S+ h(y^S) D v^S)+b(y^S) D v^{S}
\end{equation*}
by Theorem 3.7, Chapter V in \cite{shige_script}.
	By construction the canonical processes $B^S$ and $B^{\mu}$ are independent under all $P \in \tilde{\mathcal{A}}(t,y), (t,y) \in [0,T] \times \mathbb{R}^2$. We have
	\begin{align}
		v(t,y)&= \sup_{P \in \tilde{\mathcal{A}}(t,y)} \frac{ E^P[e^{-\int_t^T B^{\mu}_s ds} f(S_T)  \textbf{1}_{\lbrace B_t^{\mu}=y^{\mu} \rbrace}\textbf{1}_{\lbrace S_t=y^{S} \rbrace}]}{P(B_t^{\mu}=y^{\mu}, S_t=y^{S})} \nonumber \\
		&= \sup_{P \in \tilde{\mathcal{A}}(t,y)} \frac{ E^P[e^{-\int_t^T B^{\mu}_s ds}  \textbf{1}_{\lbrace B_t^{\mu}=y^{\mu} \rbrace}]E^P[ f(S_T) \textbf{1}_{\lbrace S_t=y^{S} \rbrace}]}{P(B_t^{\mu}=y^{\mu})P(S_t=y^{S})}  \label{IndependenceI}\\
		&=\sup_{P^{\mu} \in \mathcal{A}(t,y^{\mu}, \Theta^{\mu})} \frac{ E^{P^{\mu}}[e^{-\int_t^T B^{\mu}_s ds}  \textbf{1}_{\lbrace B_t^{\mu}=y^{\mu} \rbrace}]}{P(B_t^{\mu}=y^{\mu})} \sup_{P^S \in \mathcal{P}^S} \frac{ E^{P^{S}}[f(S_T)  \textbf{1}_{\lbrace S_t=y^{S} \rbrace}]}{P(S_t=y^{S})} \label{IndependeceII} \\
		&=\sup_{P^{\mu} \in \mathcal{A}(t,y^{\mu}, \Theta^{\mu})} E^{P^{\mu}}[e^{-\int_t^T B^{\mu}_s ds}  \textbf{1}_{\lbrace B_t^{\mu}=y^{\mu} \rbrace}] \sup_{P^S \in \mathcal{P}^S} \frac{ E^{P^{S}}[f(S_T)  \textbf{1}_{\lbrace S_t=y^{S} \rbrace}]}{P(S_t=y^{S})} \label{IndependenceIII} \\
		&=:v^{\mu}(t,y^{\mu}) v^{S}(t,y^{S}). \nonumber
\end{align}
In $(\ref{IndependenceI})$ we used the independence of $B^S$ and $B^{\mu}$. Moreover, by Definition \ref{defiAffineDominated} of $\mathcal{A}(t,y^{\mu},\Theta^{\mu})$ it holds $P^{\mu}(B_t^{\mu}=y^{\mu})=1$ for all $P^{\mu} \in \mathcal{A}(t,y^{\mu},\Theta^{\mu})$ which implies $(\ref{IndependeceII})$ and $(\ref{IndependenceIII})$.
If $\mathcal{A}^{\mu}(y^{\mu}, \Theta^x)$ satisfies the assumptions in Proposition 6.2 in \cite{fns_2019} (which corresponds to conditions in Proposition \ref{GeneralizedNumericalRepr} in a one-dimensional setting), then the function $v^{\mu}:[0,T] \times \mathcal{O}^{\mu} \to \mathbb{R}$ can be expressed via generalized Riccati equations in Proposition \ref{GeneralizedNumericalRepr}. 
\end{example}

\end{subsection}
\end{section}

\begin{section}{Pricing under model uncertainty in an arbitrage-free setting} \label{SectionPricing}

We now wish to show how the extended market model on $\tilde{\Omega}_x$ introduced in Subsection \ref{SubsectionExtendedMarketModel} containing the riskfree asset $S^0$, the risky asset $S$ and the longevity bond $S^L$ is arbitrage-free. More in general, we allow the trading of a contingent claim represented by a $\mathcal{G}_T^{{{\mathcal{Z}}}}$-measurable random variable $Y$. We price this contingent claim with the sublinear conditional operator $(\tilde{\mathcal{E}}_t)_{t \in [0,T]}$ introduced in Proposition \ref{extendedOperator}, i.e., we set $S^Y_t:=\tilde{\mathcal{E}}_t(Y)$ for $t \in [0,T]$. To guarantee that $S^Y$ is well-defined we assume from now on that $Y$ is upper semianalytic on $\tilde{\Omega}_x$ and nonnegative. We then show that the extended market model $(S^0,S,S^Y)$ is arbitrage-free. Setting $Y:=e^{-\int_0^T B_s^{\mu}ds}$ we obtain the desired result for the market model extended with the longevity bond.\\
We now consider the concept of ``absence of arbitrage of the first kind'' NA$_1(\mathcal{\tilde{\mathcal{Z}}})$ under model uncertainty introduced in \cite{bbkn_2017} and directly apply it to our market model on $\tilde{\Omega}_x$. For a $\sigma$-field $\mathcal{A} \subseteq \mathcal{G}$ the set of all $[0,\infty]$-valued, $\mathcal{G}$-measurable random variables that are $\tilde{\mathcal{Z}}$-q.s. finite is denoted by $L_+^0(\mathcal{A}, \tilde{\mathcal{Z}})$. A trading strategy $H$ is given by a simple predictable processes $H=\sum_{i=1}^n h_i \textbf{1}_{]\tau_{i-1},\tau_i]}$, where $h_i=(h_i^j)$, $j \in \lbrace S,Y\rbrace$ is $\mathcal{G}_{\tau_{i-1}}$-measurable for all $i \leq n$ and $(\tau_i)_{i \leq n}$ is a nondecreasing sequence of $\mathbb{G}$-stopping times with $\tau_0=0$. The set of possible strategies for a given initial wealth $x \in \mathbb{R}_+$ is given by
\begin{equation}
	\mathcal{H}^{\text{simp}}(x)= \lbrace H: \text{ simple predictable process such that } X^{x,H}\geq 0 \ \mathcal{\tilde{Z}} \text{-q.s.}  \rbrace, \label{StrategiesSimple}
\end{equation}
where $X^{x,H}$ is the associated wealth process of the form
\begin{equation}
	X^{x,H}_t=\underbrace{x+ \sum_{i=1}^n h_i^S (S_{\tau_i \wedge t}-S_{\tau_{i -1} \wedge t})}_{:=X^{x,H,1}_t} + \underbrace{\sum_{i=1}^n h_i^Y (S_{\tau_i \wedge t}^Y-S_{\tau_{i -1} \wedge t}^Y )}_{:=X^{x,H,2}_t}. \label{specialFormS}
\end{equation}
A simple strategy $H$ is in $\mathcal{H}^{\text{simp}}(x)$ if $X^{x,H}$ stays nonnegative $\tilde{\mathcal{Z}}$-q.s.
We introduce the set
\begin{equation}
	\mathcal{X}^{\text{simp}} = \lbrace X^{x,H} : x \in \mathbb{R}_+, \ H \in \mathcal{H}^{\text{simp}}(x) \rbrace.
\end{equation}
For $T \in \mathbb{R}_+$ and $f \in L_+^0(\mathcal{G}_T,\tilde{\mathcal{Z}})$ the superhedging price of the claim $f$ is defined by
\begin{equation}
	\nu^{\text{simp}}(T,f):= \inf \lbrace x \in \mathbb{R}_+: \exists H \in \mathcal{H}^{\text{simp}}(x) \text{ with } X_T^{x,H} \geq f \ \tilde{\mathcal{Z}} \text{-q.s.} \rbrace. \label{superhedgingSimple}
\end{equation}
\begin{defi}{\cite[Definition 2.1]{bbkn_2017}}\label{DefiNA1ExtendedMarket}
The market model $(S,S^Y)$ on $\tilde{\Omega}_x$ presents no arbitrage of first kind with respect to $\tilde{\mathcal{Z}}$, (NA$_1(\tilde{\mathcal{Z}})$) if 
\begin{equation}
	\forall s \in [0,T] \text{ and } f \in L_+^0(\mathcal{G}_s,\tilde{\mathcal{Z}}), \quad \nu^{\text{simp}}(s,f)=0 \Longrightarrow f=0 \quad \tilde{\mathcal{Z}}\text{-q.s.,} \label{extendedMarketNA}
\end{equation}
where the wealth process $X^{x,H}$ given as in $(\ref{specialFormS})$ and $\tilde{\mathcal{Z}}:=\tilde{\mathcal{Z}}(0,\omega)$ defined in $(\ref{extendedSet})$.
\end{defi}

The arbitrage condition in Definition \ref{DefiNA1ExtendedMarket} takes only into account the set $\tilde{\mathcal{Z}}$ and not the families $(\tilde{\mathcal{Z}}(t,\omega))_{(t,\omega) \in (0,T] \times \Omega_x}$ which is in line with the assumptions in Theorem 3.2 in \cite{n_2015}. This is motivated by the fact that $\tilde{\mathcal{Z}}$ is the set of probability measures we are really interested in and the families $(\tilde{\mathcal{Z}}(t,\omega))_{(t,\omega) \in (0,T] \times \Omega_x}$ are auxiliary constructions. In addition, in our setting the set $\tilde{\mathcal{Z}}(t,\omega)$ intuitively considers the market on the interval $[t,T]$ instead from time zero. 

\begin{remark} \label{RemarkFiltrationsEqual}
In contrast to \cite{bbkn_2017} we do not assume the asset $S$ to have $\tilde{\mathcal{Z}}$-q.s. continuous paths which is crucial for proving the fundamental theorem of asset pricing in Theorem 3.4 in \cite{bbkn_2017}. Here, we only require paths to be c\`{a}dl\`{a}g as in the classical case, e.g. \cite{kardaras_2010}, or without any assumptions regarding regularity. Another difference to \cite{bbkn_2017} is that the simple predictable strategies $H$ are defined with respect to the filtration $\mathbb{G}$ and not with respect to the right-continuous filtration $\mathbb{G}_+$. As already mentioned in \cite{bbkn_2017} this is not a problem as the set of predictable processes on $(\tilde{\Omega}, \mathbb{G}_+)$ coincides with the class of predictable processes on $(\tilde{\Omega}, \mathbb{G})$. Furthermore, the set of local martingale measures in Definition 3.3 in \cite{bbkn_2017} is also defined by the local martingale property with respect to $\mathbb{G}_+$. By Proposition 2.2 in \cite{nn_measurability_2014} it holds that for any right-continuous $\mathbb{G}$-adapted process it is equivalent to be a $(\tilde{P},\mathbb{G})$-semimartingale or a $(\tilde{P},\mathbb{G}_+^P)$-semimartingale or a $(\tilde{P},\mathbb{G}_+)$-semimartingale and the semimartingale characteristics are the same. Thus, it is also possible to consider local $(\tilde{P},\mathbb{G})$-martingales instead of local $(\tilde{P},\mathbb{G}_+)$-martingales.
\end{remark}
We now introduce the weaker notion NA$_1(\tilde{P}):=$ NA$_1(\lbrace \tilde{P} \rbrace)$ for $\tilde{P} \in \tilde{\mathcal{Z}}$ which is used in the proof of \cite[Theorem 3.4]{bbkn_2017}. This condition means that
\begin{equation*}
	\forall T \in \mathbb{R}_+ \text{ and } f \in L_+^0(\mathcal{G}_T,\tilde{P}), \quad \nu^{\text{simp},\tilde{P}}(T,f)=0 \Longrightarrow f=0 \ \tilde{P}\text{-a.s.},
\end{equation*}
where 
\begin{equation*}
	\nu^{\text{simp},\tilde{P}}(T,f):= \inf \lbrace x \in \mathbb{R}_+: \exists H \in \mathcal{H}^{\text{simp},\tilde{P}}(x) \text{ with } X_T^{x,H} \geq f \ \tilde{P} \text{-a.s.} \rbrace
\end{equation*}
and $\mathcal{H}^{\text{simp},\tilde{P}}(x)$ is the class of all simple predictable processes such that $X^{x,H}$ is non-negative $\tilde{P}$-a.s. We have the following useful relation between NA$_1(\tilde{\mathcal{P}})$ and NA$_1(\tilde{P})$. 
	
\begin{prop} \label{NA_prop_equivalence}
Assume that $S$ has $\tilde{\mathcal{Z}}$-q.s. continuous paths. Then
\begin{equation} \label{NA_equivalence}
	\text{NA}_1(\tilde{\mathcal{Z}}) \text{ holds if and only if } \text{NA}_1(\tilde{P}) \text{ holds for all } \tilde{P} \in \tilde{\mathcal{Z}}. 
\end{equation}
If $S$ has c\`{a}dl\`{a}g paths, then
\begin{equation} \label{NA_one_direction}
	 \text{NA}_1(\tilde{P}) \text{ holds for all } \tilde{P} \in \tilde{\mathcal{Z}} \text{ implies } \text{NA}_1(\tilde{\mathcal{Z}}). 
\end{equation}
\end{prop}
\begin{proof}
Equivalence $(\ref{NA_equivalence})$ follows by Theorem 3.4 in \cite{bbkn_2017}. As $\mathcal{H}^{\text{simp}} (x) \subseteq \mathcal{H}^{\text{simp},P}(x)$ we have that $(\ref{NA_one_direction})$ holds.
\end{proof}
\begin{remark}
The other direction in $(\ref{NA_one_direction})$ relies on the property of $\tilde{\mathcal{Z}}$-q.s. continuous paths of $S$ and does not hold in general. 
\end{remark}	
The following lemma shows that the original market model $(S^0,S)$ on $(\Omega_x,\mathcal{F}_T)$ satisfies NA$_1(\mathcal{Z})$ in the sense of Definition $\ref{DefiNA1ExtendedMarket}$ by considering $\mathcal{Z}$ and the filtration $\mathbb{F}$ instead of $\tilde{\mathcal{Z}}$ and $\mathbb{G}$ respectively. Moreover, the wealth process $X^{x,H}$ defined in $(\ref{specialFormS})$ consists only of $X_t^{x,H,1}$, i.e., $S^Y \equiv 0$. 
\begin{lemma} \label{StartMarketNA}
	Under Assumption \ref{AssumptionMartingale} the condition NA$_1(\mathcal{Z})$ is satisfied for the market model $(S^0,S)$ on $(\Omega_x, \mathcal{F}_T)$ defined in Subsection \ref{SubsectionMarketModel}.
\end{lemma}
	
\begin{proof}
 Assumption \ref{AssumptionMartingale} ensures that for every $Q \in \mathcal{Z}$ the NFLVR-condition holds for the $Q$-market. This implies that NA$_1(Q)$ holds for all $Q \in \mathcal{Z}$ and by Proposition \ref{NA_prop_equivalence} we can conclude the NA$_1(\mathcal{Z})$ holds. Here, we used that in the classical case, i.e., when the set of priors consists only of one single probability measure, it holds that NFLVR implies NA$_1$ by Lemma A.2 in \cite{kabanov_kardaras_song_2016}.
\end{proof}

In the sequel, we prove that no arbitrage of first kind under model uncertainty also holds for extended models $(S^0,S,S^Y)$ on $(\tilde{\Omega}_x, \mathcal{G}_T)$. 

\begin{asum} \label{asumExpectationD}
	Let $\tilde{P} \in \tilde{\mathcal{Z}}$ and $t \in [0,T]$. Then for all $X^{x,H} \in \mathcal{X}^{\text{simp}}$ we have
	\begin{equation}
		E^{\tilde{P}} [X_t^{x,H}] \leq E^{\tilde{P}} [X_0^{x,H}]\label{ExpectationD}.
	\end{equation}
\end{asum}

\begin{prop} \label{NADecrasing} Let $Y$ be an upper semianalytic, $\mathcal{G}_T^{{{\mathcal{Z}}}}$-measurable and nonnegative random variable. Set $S_t^Y:= \tilde{\mathcal{E}}_t(Y)$ for $t \in [0,T]$. Under Assumption \ref{asumExpectationD} the extended market model $(S^0,S,S^Y)$ on $\tilde{\Omega}_x$ satisfies NA$_1(\tilde{\mathcal{Z}})$. 
\end{prop}
\begin{proof}
	This follows by the arguments in the proof of Theorem 3.5 in \cite{bbkn_2017}.
\end{proof}
Assumption \ref{asumExpectationD} may appear restrictive. However, it is satisfied in many cases, as we now show below. 
\begin{lemma} \label{LemmaExpectationSatisfied}
If one of the following properties holds for every $X^{x,H} \in \mathcal{X}^{\text{simp}}$, then condition $(\ref{ExpectationD})$ is satisfied. 
	\begin{enumerate}
	\itemsep0pt
		\item $X^{x,H}$ is a $\tilde{P}$-supermartingale for all $\tilde{P} \in \tilde{\mathcal{Z}}$. 
		\item $X^{x,H,1}$ and $X^{x,H,2}$ are $\tilde{P}$-supermartingales for all $\tilde{P} \in \tilde{\mathcal{Z}}$. 
		\item $X^{x,H,1} \geq 0$ $\tilde{\mathcal{Z}}$-q.s. and $ E^{\tilde{P}} [ X^{x,H,2}_{t}] \leq 0$ for all $\tilde{P} \in \tilde{\mathcal{Z}}$ and $t \in [0,T]$.
		\item $X^{x,H,1} \geq 0$ $\tilde{\mathcal{Z}}$-q.s. and we do not allow short-selling for $S^Y$, i.e., $h_i^{Y} \geq 0$ for $i=1,..,n$.
		\item $S \geq 0$ $\tilde{\mathcal{Z}}$-q.s. and we do not allow short-selling for $S,S^Y$, i.e., $h_i^{j} \geq 0$ for $i=1,..,n$ and $j \in \lbrace S,Y \rbrace$.
	\end{enumerate}
Clearly, $2. \Rightarrow 1.$ and $5. \Rightarrow 4. \Rightarrow 3.$. 
\end{lemma}

\begin{proof}
Since $(\ref{ExpectationD})$ obviously holds under conditions 1. and 2., we start with 3. Let $X^{x,H} \in \mathcal{X}^{\text{simp}}$. 
As for any $\tilde{P} \in \tilde{\mathcal{Z}}$ the asset $S$ is a $(\tilde{P}, \mathbb{G})$-local martingale, there exists an increasing sequence $(\tilde{\tau}_n)_{n \in \mathbb{N}}$ of $\mathbb{G}$-stopping times with $\tilde{\tau}_n \uparrow \infty$ $\tilde{P}$-a.s. such that $(S_{\tilde{\tau}_n \wedge t})_{t \geq 0}$ is a $\tilde{P}$-martingale for all $\tilde{P} \in \tilde{\mathcal{Z}}, n \in \mathbb{N}$. It follows that for $ x \in \mathbb{R}_+,  H \in \mathcal{H}^{\text{simp}}(x), n \in \mathbb{N}$
\begin{equation} 
		X^{x,H,1}_{\cdot \wedge \tilde{\tau}_n} \text{ is a local } (\tilde{P}, \mathbb{G}) \text{-martingale for all } \tilde{P} \in \tilde{\mathcal{Z}}. \label{localmartingale}
\end{equation}
By $X^{x,H,1} >0$ $\tilde{\mathcal{Z}}$-q.s., $X^{x,H,1}_{\cdot \wedge \tilde{\tau}_n}$ is a $\tilde{P}$-supermartingale for all $\tilde{P} \in \tilde{\mathcal{Z}}$ and $n \in \mathbb{N}$. Thus, by Fatou's Lemma we get for $0 \leq s \leq t$ and $\tilde{P} \in \tilde{\mathcal{Z}}$
\begin{equation*}
	 E^{\tilde{P}}[X_t^{x,H,1} \vert \mathcal{G}_s] = E^{\tilde{P}}\big[\lim_{n \to \infty} X_{t \wedge \tilde{\tau}_n}^{x,H,1} \vert \mathcal{G}_s\big] \leq \liminf_{n \to \infty} E^{\tilde{P}} \big[X_{t \wedge \tilde{\tau}_n}^{x,H,1} \vert \mathcal{G}_s \big] \leq \liminf_{n \to \infty} X_{s \wedge \tilde{\tau}_n}^{x,H,1}=X_s^{x,H,1},
	 \end{equation*}
i.e., $X^{x,H,1}$ is a $\tilde{P}$-supermartingale for all $\tilde{P} \in \tilde{\mathcal{Z}}$.	
As $E^{\tilde{P}} [ X^{x,H,2}_{t}] \leq 0$ for all $\tilde{P} \in \tilde{\mathcal{Z}}$ and $t \in [0,T]$ it follows for $\tilde{P} \in \tilde{\mathcal{Z}}$
\begin{align}
	E^{\tilde{P}}[X_t^{x,H}] &=  E^{\tilde{P}} [ X^{x,H,1}_{t}]  +E^{\tilde{P}}[ X^{x,H,2}_{t}] \nonumber \leq E^{\tilde{P}} [X^{x,H,1}_0] \nonumber \\
	&= E^{\tilde{P}} \bigg[X^{x,H,1}_0 + \sum_{i=1}^n h_i^{Y} \underbrace{(S_{\tau_i \wedge 0}^{Y}-S_{\tau_{i-1} \wedge 0}^{Y})}_{=0} \bigg] = E^{\tilde{P}}[X_0^{x,H}]. \nonumber 
\end{align}
For condition 4. it is enough to observe that
\begin{equation*}
	E^{\tilde{P}} [ X^{x,H,2}_{t} ]= \sum_{i=1}^n h_i^{Y} \underbrace{\big(E^{\tilde{P}} [S_{\tau_i \wedge t }^{Y}]-E^{\tilde{P}} [S_{\tau_{i-1} \wedge t }^{Y}] \big)}_{\leq 0}\leq 0,
\end{equation*}
by using the no short-sale constraint and the fact that $S^Y$ is a supermartingale for all $\tilde{P} \in \tilde{\mathcal{Z}}$. \\
Furthermore, it is obvious that the nonnegativity of $S$ and the additional short-sale constraint guarantee that $X^{x,H,1} \geq 0$ $ \tilde{\mathcal{Z}}$-q.s..
\end{proof}
\begin{remark}
In condition 1. and 2. in Lemma \ref{LemmaExpectationSatisfied} the chosen filtrations play no role as $(\ref{ExpectationD})$ involves only the expectation. 
\end{remark}
The results in Lemma \ref{LemmaExpectationSatisfied} show that the sublinear conditional operator $(\tilde{\mathcal{E}}_t)_{t \in [0,T]}$ allows to price a European contingent claim in a way that the extended market is arbitrage-free way under some additional assumptions. These supplementary constraints can be regarded as the price we pay for considering a setting under model uncertainty.
On the one hand, allowing only strategies $H$ such that the wealth process $X^{x,H}$ is a supermartingale for all $\tilde{P} \in \tilde{\mathcal{Z}}$ is in line with the definition of admissible strategies under model uncertainty in \cite[p. 4450]{n_2015} with the difference that there not only simple strategies are considered. 
On the other hand, the supermartingale assumption seems too strong due to Assumption \ref{asumExpectationD} which only requires decreasing expectation. Conditions 3. and 4. in Lemma \ref{LemmaExpectationSatisfied} could be regarded as restrictive in an economical sense. However, in an insurance context constraints as no short-selling or a positive wealth-process are often required by the regulatory framework.\\
For the next result we consider general families of probability measures $({\mathcal{P}}(t,\omega))_{(t,\omega) \in [0,T] \times \Omega}$ on $\Omega$ satisfying Assumption \ref{assumptionnutzNew} in the setting of Section \ref{Section_Reduced_framework}.
\begin{lemma} \label{lemmaAnotherExample}
Let Assumption \ref{assumptionnutzNew} hold for $({\mathcal{P}}(t,\omega))_{(t,\omega) \in [0,T] \times \Omega}$ and $Y$ be an upper semianalytic, $\mathcal{G}_T^{{{\mathcal{P}}}}$-measurable and nonnegative function on $\tilde{\Omega}$. Set $S_t^Y:= \textbf{1}_{N^c} \limsup_{r \downarrow t, r \in \mathbb{Q}} \tilde{\mathcal{E}}_r(Y)$ for $t \in [0,T)$ and $S_T^Y:=\tilde{\mathcal{E}}_T(Y)$ with $N \in \mathcal{N}_T^{\tilde{P}}$. Let $S$ be an $\mathbb{G}^{*,{\tilde{\mathcal{P}}}}$-adapted continuous $(\tilde{P},\mathbb{G}^{*,{\tilde{\mathcal{P}}}})$-semimartingale\footnote{By the same arguments regarding the filtration as in Remark \ref{RemarkFiltrationsEqual} $S$ is also $(\tilde{P},\mathbb{G}^{*,{\tilde{\mathcal{P}}}}_+)$-semimartingale for all $\tilde{P} \in \tilde{\mathcal{P}}$.} for all $\tilde{P} \in \tilde{\mathcal{P}}$.
Assume that $\tilde{\mathcal{P}}$ is a non-empty saturated set of sigma-martingale measures\footnote{The sigma-martingale property holds with respect to the filtration $\mathbb{G}^{*,{\tilde{\mathcal{P}}}}_+$ for all $\tilde{P} \in \tilde{\mathcal{P}}$.} for $S$.	
Under the assumption $h_i^Y \geq 0, i=1,...n,$ the extended market model $(S^0,S,S^Y)$ satisfies $NA(\mathcal{\tilde{P}})$ for $\tilde{\mathcal{Z}}=\tilde{\mathcal{P}}$ in $(\ref{extendedMarketNA})$ .
\end{lemma}
	
\begin{proof} 
By applying Theorem 2.4 in \cite{n_2015}, a version of the optional decomposition theorem under model uncertainty, for every $\tilde{P} \in \tilde{\mathcal{P}}$ there exists a $\mathbb{G}_+^{*,\tilde{\mathcal{P}}}$-predictable process $\tilde{H}$ which is $S$-integrable for all $\tilde{P} \in \tilde{\mathcal{P}}$ such that 
\begin{equation}
	D_t:=S_t^Y - \int_0^{t,(\tilde{P})} H  dS \text{ is nonincreasing } \tilde{P} \text{-a.s. for all } \tilde{P} \in \tilde{\mathcal{P}}. \label{optionalDecomp}
\end{equation}
In this case $\int_0^{t,(\tilde{P})} H  dS$ is the Itô-integral under the fixed measure $\tilde{P} \in \tilde{\mathcal{P}}$. As $S$ is a continuous local $(\tilde{P},\mathbb{G}^{*,{\tilde{\mathcal{P}}}}_+)$-martingale for all $\tilde{P} \in \tilde{\mathcal{P}}$ and $H$ is integrable, it follows that also $\int_0^{t,(\tilde{P})} H  dS$ is a continuous local $(\tilde{P},\mathbb{G}^{*,{\tilde{\mathcal{P}}}}_+)$-martingale for all $\tilde{P} \in \tilde{\mathcal{P}}$.
Here, we use that a sigma-martingale with continuous paths is a local martingale by Theorem 91 (IV.9) in \cite{protter}.
Consider now $X^{x,H} \in \mathcal{X}^{\text{simp}}$. By $(\ref{optionalDecomp})$ we get 
\begin{align*}
	0 \leq X^{x,H}_t &= X_t^{x,H,1}+ \sum_{i=1}^n h_i^Y (S_{\tau_i \wedge t}^Y-S_{\tau_{i -1} \wedge t}^Y ) \\
	&= X_t^{x,H,1}+ \sum_{i=1}^n h_i^Y \bigg(D_{\tau_i \wedge t}+ \int_0^{\tau_i \wedge t,(\tilde{P})} H  dS-D_{\tau_{i-1} \wedge t} - \int_0^{\tau_{i-1} \wedge t,(\tilde{P})} H  dS \bigg)\\
	&=  X_t^{x,H,1}+ \sum_{i=1}^n h_i^Y \underbrace{(D_{\tau_i \wedge t}-D_{\tau_{i-1} \wedge t})}_{\leq 0} + \sum_{i=1}^n h_i^Y \bigg ( \int_{\tau_{i-1} \wedge t}^{\tau_i \wedge t,(\tilde{P})} H  dS \bigg) \\
	& \leq  X_t^{x,H,1} + \sum_{i=1}^n h_i^Y \bigg ( \int_{\tau_{i-1} \wedge t}^{\tau_i \wedge t,(\tilde{P})} H  dS \bigg) :=\tilde{X}_t^{x,H}.
\end{align*}
	Note, $\tilde{X}^{x,H}_{\cdot \wedge \tilde{\tau}_m}$ is a $(\tilde{P},\mathbb{G}^{*,{\tilde{\mathcal{P}}}}_+)$-supermartingale for all $\tilde{P} \in \tilde{\mathcal{P}}$, $m \in \mathbb{N}$. 
	It follows by Fatou's Lemma that also $\tilde{X}^{x,H}$ is a $(\tilde{P},\mathbb{G}^{*,{\tilde{\mathcal{P}}}}_+)$-supermartingale for all $\tilde{P} \in \tilde{\mathcal{P}}$. Thus, for all $t \in [0,T]$
	\begin{align*}
		E^{\tilde{P}}[X_t^{x,H}] \leq E^{\tilde{P}}[\tilde{X}^{x,H}_t] \leq E^{\tilde{P}}[\tilde{X}^{x,H}_0] = E^{\tilde{P}}\big[X_0^{x,H}+\sum_{i=1}^n h_i^Y \underbrace{(D_{\tau_i \wedge 0}-D_{\tau_{i-1} \wedge 0})}_{= 0} \big] = E^{\tilde{P}}[X_0^{x,H}] 
	\end{align*}
	for $\tilde{P} \in \tilde{\mathcal{P}}$ which implies Assumption \ref{asumExpectationD}.
\end{proof}		

Here, we only assume a no short-selling constraint for the strategies. The price we pay for this are more assumptions on $\mathcal{P}$. An example for a set of probability measures satisfying these conditions is given in Lemma 4.2 in \cite{n_2015}. However, the set $\tilde{\mathcal{Z}}$ does not satisfy these assumptions as already a set of affine processes is not saturated even under one single prior. In general, the optional decomposition theorem under model uncertainty in \cite{n_2015} also requires that $S$ has non-dominating diffusions under each $\tilde{P} \in \tilde{\mathcal{P}}$. However, this property is always satisfied if $S$ is continuous, see Example 2.3 ii) in \cite{n_2015}.\\
As already mentioned the set $\mathcal{Z}$ is in general not saturated due to the affine structure, as it is outlined in the following. Let $P \in \mathcal{Z}$ such that $S$ is a positive local $(P,\mathbb{F}_+^{*,\mathcal{Z}})$-martingale. Consider $P' \in \mathcal{P}(\Omega_x)$ such that $P \sim P'$ and $B^{S}$ is a local $(P',\mathbb{F}_+^{*,\mathcal{Z}})$-martingale. 
By the definition of $\mathcal{Z}$, $B^{S}$ is $\mathbb{F}$-adapted and thus a local $(P',\mathbb{F})$-martingale by Theorem 10 in \cite{foellmer_protter_2010}. Furthermore, as $B^{\mu}$ is a $(P, \mathbb{F})$-semimartingale, it follows that $B^{\mu}$ is a $(P'$,$\mathbb{F})$-semimartingale due to $P \sim P'$ by Theorem III.3.13 in \cite{js_2013}. By applying Girsanov's theorem for semimartingales in Proposition III.3.24 in \cite{js_2013} to $B^{\mu}$, there exists a predictable process $b$ satisfying
	\begin{equation*}
		\int |\alpha_s b_s| ds < \infty \text{ and } \int b_s^2 \alpha_s ds < \infty \ P' \text{-a.s. for } t \in [0,T]
	\end{equation*}
	and such that a version of the characteristics of $B^{\mu}$ relative to $P'$ is given by
	\begin{equation}
		A^{P'} = A^{P} + \int \alpha_s b_s ds =\int \underbrace{\big (\beta^P_s + \alpha_s b_s \big)}_{:=\beta^{P'}_s} ds,  \quad  C^{P'}=C \label{girsanov}
	\end{equation}
	up to a $P'$-null set. By $(\ref{girsanov})$ we can see why the saturation property is not satisfied for an arbitrary affine structure in Definition \ref{defiAffineDominated}, as we can not guarantee $\beta_s^{P'} \in b^*(B^{\mu}_s)$ for $dP' \otimes dt$-almost all $(s,\omega) \in \Omega_x \times [t,T]$. However, by considering only an affine structure on the volatility of the mortality intensity, as it is the case in Example \ref{Example2RandomG-expectation}, the set $\mathcal{Z}$ is saturated.\\		
We now compare the price process $(\tilde{\mathcal{E}}_t(Y))_{t \in [0,T]}$ of the contingent claim $Y$ with its corresponding superhedging price. In this setting Theorem 3.11 and 3.12 in \cite{bz_2019} can be reformulated if 
$\tilde{\mathcal{P}}$ satisfies Assumption 3.1 in \cite{bz_2019}.
\begin{asum} \label{sigmaSaturated} \
\begin{enumerate} 
\itemsep0pt
	\item $\tilde{\mathcal{P}}$ is a set of sigma martingale measures for $S$, i.e., $S$ is a $(\tilde{P}, \mathbb{G}_+^{*,\mathcal{\tilde{P}}})$-sigma-martingale for all $\tilde{P} \in \tilde{\mathcal{P}}$;
	\item $\tilde{\mathcal{P}}$ is saturated: all equivalent sigma-martingale measures of its elements still belong to $\tilde{\mathcal{P}}$;
	\item $S$ has dominating diffusion under every $\tilde{P} \in \tilde{\mathcal{P}}$.
\end{enumerate}
\end{asum}
Here, $S$ is assumed to be a $d$-dimensional $\mathbb{G}^{*,\tilde{\mathcal{P}}}$-adapted process with c\`{a}d\l\`{a}g paths such that $S$ is a $(\tilde{P}, \mathbb{G}^{*,\tilde{\mathcal{P}}})$-semimartingale for every $\tilde{P} \in \tilde{\mathcal{P}}$. 
Furthermore, the set of $d$-dimensional $\mathbb{G}^{*,\tilde{\mathcal{P}}}$-predictable processes which are $S$-integrable for all $\tilde{P} \in \tilde{\mathcal{P}}$ is denoted by $\tilde{L}(S,\tilde{\mathcal{P}})$ and the admissible strategies on $\tilde{\Omega}$ are given by
\begin{equation*}
\tilde{\bigtriangleup}:=\bigg\lbrace \tilde{\delta} \in \tilde{L}(S,\tilde{\mathcal{P}}): \int^{(\tilde{P})} \tilde{\delta} dS \text{ is a } (\tilde{P},\mathbb{G}_+^{*,\tilde{\mathcal{P}}}) \text{-supermartingale for all } \tilde{P} \in \tilde{\mathcal{P}} \bigg\rbrace.
\end{equation*}
In this case the notation $\int^{(\tilde{P})} \tilde{\delta} dS := (\int^{(\tilde{P}),t} \tilde{\delta} dS)_{t \in [0,T]}$ is the usual It\^{o} integral under $\tilde{P}$. 
We recall \cite[Theorem 3.11]{bz_2019}.
\begin{theorem} \label{superhedging}
Let Assumption \ref{assumptionnutzNew} hold for $(\mathcal{P}(t,\omega))_{(t, \omega) \in [0,T] \times \Omega}$ and Assumption \ref{sigmaSaturated} for $\tilde{\mathcal{P}}$, respectively. Consider $Y$ to be an upper semianalytic, $\mathcal{G}_T^{\mathcal{P}}$-measurable and nonnegative contingent claim such that $\tilde{\mathcal{E}}_t(Y) \in L^1(\tilde{\Omega})$ for all $t \in [0,T]$. 
If $t \in [0,T]$ and there exists a $\mathbb{G}^{*,\tilde{\mathcal{P}}}$-adapted process $\tilde{X}=(\tilde{X}_s)_{s \in [0,T]}$ with c\`{a}dl\`{a}g paths, such that for $s \in [0,T]$
\begin{equation*}
	\tilde{X}_s = \tilde{\mathcal{E}}_s(Y) \quad \tilde{P}\text{-a.s.} \text{ for all } \tilde{P} \in \tilde{\mathcal{P}},
\end{equation*}
and if the tower property holds for $Y$, i.e., for all $r,s \in [0,t]$ with  $r \leq s$,
\begin{equation*}
	\tilde{\mathcal{E}}_t(Y)=\tilde{\mathcal{E}}_r(\tilde{\mathcal{E}}_s(Y)) \quad \tilde{P}\text{-a.s.} \text{ for all } \tilde{P} \in \tilde{\mathcal{P}},
\end{equation*}
then we have the following equivalent dualities for all $\tilde{P} \in \tilde{\mathcal{P}}$ and $s \in [0,T]$
\begin{align}
	\tilde{\mathcal{E}}_s(Y) 
	=& \essinf \lbrace \tilde{v} \text{ is } \mathcal{G}_s^{*,\tilde{\mathcal{P}}}\text{-measurable}: \exists \tilde{\delta} \in \tilde{\bigtriangleup} \text{ such that } \tilde{v} + \int_s^{(\tilde{P}'),T} \tilde{\delta}_u dS_u \geq Y \quad \tilde{P}' \text{-a.s. } \nonumber \\ 
	&\text{ for all } \tilde{P}' \in \tilde{\mathcal{P}} \rbrace =:\essinf \lbrace D_s \rbrace \quad \tilde{P} \text{-a.s.}  \label{duality1} \\
	=& \essinf \lbrace \tilde{v} \text{ is } \mathcal{G}_s^{*,\tilde{\mathcal{P}}}\text{-measurable}: \exists \tilde{\delta} \in \tilde{\bigtriangleup} \text{ such that } \tilde{v} + \int_s^{(\tilde{P}'),T} \tilde{\delta}_u dS_u \geq Y \quad \tilde{P}' \text{-a.s. } \nonumber \\ 
	&\text{ for all } \tilde{P}' \in \tilde{\mathcal{P}}(s;\tilde{P}) \rbrace =:\essinf \lbrace D_s^{(\tilde{P})} \rbrace \quad \tilde{P} \text{-a.s.} \label{duality2}
\end{align}
\end{theorem}
If $\tilde{\mathcal{P}}$ does not satisfy Assumption \ref{sigmaSaturated}, the superhedging dualities $(\ref{duality1}), (\ref{duality2})$ do not hold in general.
However, by having a look at the proof of Theorem 5.2.21 in \cite{zhang_phd} one of the two inequalities is still valid, i.e., for $s \in [0,T]$ and $\tilde{P} \in \tilde{\mathcal{P}}$
	\begin{equation}
		\tilde{\mathcal{E}}_s(Y) \leq \essinf \lbrace D_s^{(\tilde{P})} \rbrace \leq \essinf \lbrace D_s \rbrace \quad \tilde{P} \text{-a.s.} \label{dualityInequality}
	\end{equation}
for $D_s^{(\tilde{P})}$ given in $(\ref{duality2})$ and $D_s$ in $(\ref{duality1})$.  Note, in this context $\essinf \lbrace D_s^{(\tilde{P})} \rbrace$ corresponds to the $\tilde{P}$-superhedging price and $\essinf \lbrace D_s \rbrace $ to the $\tilde{\mathcal{P}}(s;\tilde{P})$-superhedging price. The inverse inequality is not valid as an optional decomposition result for semimartingales given by Theorem 2.4 in \cite{n_2015} is used which requires Assumption \ref{sigmaSaturated} to hold. 
Nevertheless, even if the price of the contingent claim $Y$ is lower than the superhedging price, the extended market $(S^0, S, \tilde{\mathcal{E}}(Y))$ on $\tilde{\Omega}$ is still arbitrage-free under reasonable assumptions, as shown in Proposition \ref{NADecrasing}. Moreover, even under one single prior the superhedging price of a contingent claim is often criticized as being too high. In the case of considering several priors, we see in $(\ref{dualityInequality})$ that the quasi-sure $\tilde{\mathcal{P}}(s;\tilde{P})$-superhedging price is even more conservative than the one under a single prior. 
\begin{remark}
In Definition \ref{DefiNA1ExtendedMarket} of NA$_1(\mathcal{P})$ we only allow for simple trading strategies. However, in $(\ref{duality1}), (\ref{duality2})$ the strategy $\tilde{\delta} \in \tilde{\Delta}$ does not need to be simple. The same situation occurs in the superreplication result in Theorem 5.1 in \cite{bbkn_2017}. However, in the setting analyzed in Section \ref{SectionPricing} it holds 
\begin{equation}
	\tilde{\Delta}^{\text{simp}}:= \lbrace{ H \in \mathcal{H}^{\text{simp}}: X^{0,H} \text{ is a } (\tilde{P}, \mathbb{G}^{*,\tilde{\mathcal{P}}}_+) \text{-supermartingale for all } \tilde{P} \in \tilde{\mathcal{P}} \rbrace} \subseteq \tilde{\Delta},
\end{equation}
where $X^{0,H}$ is defined as $X^{0,H,1}$ in $(\ref{specialFormS})$. Thus, it follows
%\begin{equation*}
$	D_s^{\text{simp}} \subseteq D_{s}$
%\end{equation*}
with
\begin{equation*}
	D^{\text{simp}}_s:=\lbrace \tilde{v} \text{ is } \mathcal{G}_s^{*,\tilde{\mathcal{P}}}\text{-measurable}: \exists \tilde{\delta} \in \tilde{\bigtriangleup} \text{ such that } \tilde{v} + \int_s^{(\tilde{P}'),T} \tilde{\delta}_u dS_u \geq Y \ \tilde{P}' \text{-a.s. for all } \tilde{P}' \in \tilde{\mathcal{P}}(s;\tilde{P}) \rbrace.
\end{equation*}
This together with $(\ref{dualityInequality})$ implies that for all $s \in [0,T]$ and $\tilde{P} \in \tilde{\mathcal{P}}$
\begin{equation*}
	\tilde{\mathcal{E}}_s(Y) \leq \essinf \lbrace D_s \rbrace \leq \essinf \lbrace D_s^{\text{simp}} \rbrace \quad \tilde{P} \text{-a.s.}.
\end{equation*}
\end{remark} 
\end{section}

\begin{section}{Conclusion}
In this paper we were able to define an extended market model within a reduced-form framework under model uncertainty where the mortality intensity follows a non-linear affine price process. This allows both to introduce the definition of a longevity bond under model uncertainty as well as to compute it by explicit formulas or by numerical methods. We are also able to guarantee the existence of a c\`{a}dl\`{a}g  modification for the longevity bond's value process. Furthermore, we show how the resulting market model extended with the longevity bond is arbitrage-free. These results can be used for further research on hedging under model uncertainty.
\end{section}

\begin{section}{Appendix}

\begin{proof}{Proposition \ref{crucialPoint} \\} 
The proof consists in verifying that the arguments in \cite{nn_measurability_2014} are also valid in our setting. We show this explicitly how this is done for one property. The other properties follow with similar arguments. \\
\textbf{Step 1:} The families $(\mathcal{A}(t,\omega_t,\Theta))_{(t,\omega) \in [0,T] \times \Omega_x^1}$ defined in Definition \ref{defiAffineDominated} satisfy Assumption \ref{assumptionnutzNew}. 
	 1) \emph{Measurability:} From \cite[Lemma 3.1]{fns_2019} we know that the set 
\begin{equation}
	\lbrace (\omega,t,P)\in \Omega_x^1 \times [0,T] \times \mathcal{P}(\Omega_x^1) | P \in \mathcal{A}(t,\omega_t,\Theta)\rbrace \label{measurability}
\end{equation}
is Borel which implies that the set is analytic.\\
2) \emph{Invariance:}
Let $(s,\overline{\omega}) \in [0,T] \times \Omega_x^1$, $P \in \mathcal{A}(s,\overline{\omega}_s,\Theta)$ and $\tau$ be a stopping time taking values in $[s,T]$. It is clear that for every $\omega \in \Omega_x^1$ we can define the conditional probability $P^{\tau,\omega}$ with respect to $\mathcal{F}_{\tau}$ as in $(\ref{condProbability})$. We have to prove that $P^{\tau,\omega} \in \mathcal{A}(\tau, \omega_{\tau(\omega)}, \Theta)$ for $P$-a.e. $\omega \in \Omega_x^1$, i.e.,
\begin{enumerate}
\itemsep0pt
	\item $P^{\tau, \omega} \in \mathcal{P}^{ac}_{sem}$
	\item $P^{\tau,\omega}(B_{\tau(\omega)}=\omega_{\tau(\omega)})=1$
	\item $\beta^{P^{\tau,\omega}}_u \in b^*(B_u)$ and $\alpha^{P^{\tau,\omega}}_u \in a^*(B_u)$ for $dP^{\tau, \omega} \otimes dt$-almost all $(\tilde{\omega},u) \in \Omega_x^1 \times (s,T]$, 
\end{enumerate}
for $P$-a.e. $\omega \in \Omega_x^1$.
Here we denote by $\beta^{P^{\tau,\omega}}=(\beta_u^{P^{\tau,\omega}})_{u \in (s,T]}$ the absolutely continuous differential process with respect to the probability measure $P^{\tau, \omega}$ and the filtration $\mathbb{F}$. The same notation is used for the differential process $\alpha$. The second point follows directly by the definition of the probability $P^{\tau, \omega}$ in $(\ref{condProbability})$, as it is possible to choose the probability measure $P^{\tau, \omega}$ concentrated on the paths which coincide with $\omega$ up to time $\tau(\omega)$, see Section 2.1 in \cite{nh_2013}.
 The first point is a consequence of Theorem 3.1 in \cite{nn_levy_2016}, which contains two main results. First, given a probability measure $P \in \mathcal{P}^{ac}_{sem}$ it follows that for $P$-a.e. $\omega \in \Omega$ we have $P^{\tau, \omega} \in \mathcal{P}^{ac}_{sem}$. Second, given the differential characteristics of the canonical process under $P \in \mathcal{P}^{ac}_{sem}$ with respect to $\mathbb{F}$ are $(\beta^P, \alpha^P)$, then the $P^{\tau,\omega}$-$\mathbb{F}$-characteristics are given by
	\begin{equation} \label{shiftedCharacteristics}
		(\beta^{P^{\tau,\omega}}_{u}, \alpha^{P^{\tau,\omega}}_{u}):=((\beta^P_{\tau+ u})^{\tau,\omega}, \ (\alpha^P_{\tau + u})^{\tau, \omega}),
	\end{equation}
	where the notation introduced in $(\ref{concaRV})$ is used. As in our setting
	$P \in \mathcal{A}(s,\omega_s,\Theta)$, it holds  $\beta^{P}_u \in b^*(B_u)$ for $dP \otimes dt$-almost all $(\tilde{\omega},u) \in \Omega_x^1 \times (s,T]$ by the definition of the set $ \mathcal{A}(s,\omega_s,\Theta)$. This allows to conclude that
	\begin{equation*}
		\beta^{P^{\tau,\omega}}_u(\cdot) = \beta_{\tau + u}^{P} (\omega \otimes_{\tau} \cdot) \in b^*(B_{\tau+u}(\omega \otimes_{\tau} \cdot)) \text{ for } dP \otimes dt \text{-almost all } (\tilde{\omega},u) \in \Omega_x^1 \times [0,T].
	\end{equation*}
	As a consequence it holds $B_{\tau+ u}(\omega \otimes_{\tau} \cdot)=B_u(\cdot)$ $P^{\tau,\omega}$-a.s. by using $(\ref{concatenation})$ and $(\ref{condProbability})$, which proves the affine property.
	With the same arguments the result follows for the process $\alpha_s^{P^{\tau,\omega}}$.\\
3) \emph{Stability under Pasting:} 
By using similar arguments as in the proof of the invariance condition this property follows by generalizing the results of Proposition 4.1 in \cite{nn_levy_2016}.

\textbf{Step 2:} The families $(\mathcal{A}^{\mu}(t,\omega_t^{\mu},\Theta^{\mu}))_{(t,\omega) \in [0,T] \times \Omega_x}$ defined in Definition \ref{IntensityP} satisfy Assumption \ref{assumptionnutzNew}.\\ 
This follows as \cite{nn_measurability_2014} considers a $d$-dimensional setting such that the canonical process is a $d$-dimensional semimartingale process. Note, in our setting we are only interested in the structure of one component of the canonical process and thus we can apply the results in \cite{nn_measurability_2014}.
\end{proof}

\end{section}

\bibliography{Biagini_Oberpriller_Reduced_form_setting_non_linear_affine_processes.bib}{}
\bibliographystyle{plain}

\end{document}